\documentclass[10pt]{article}

\usepackage{amsmath}
\usepackage{amssymb}
\usepackage{amsthm}
\usepackage{graphics}
\usepackage[latin1]{inputenc}
\usepackage[english]{babel}
\usepackage[T1]{fontenc}

\usepackage{enumerate}

\newtheorem{thm}{Theorem}[section]
\newtheorem{prop}[thm]{Proposition}
\newtheorem{coro}[thm]{Corollary}
\newtheorem{lemma}[thm]{Lemma}
\newtheorem{rem}[thm]{Remark}
\newtheorem{defi}[thm]{Definition}

\newcommand{\R}{\mathbb{R}}             
\newcommand{\N}{\mathbb{N}}             
\newcommand{\Z}{\mathbb{Z}}             
\newcommand{\C}{\mathbb{C}}             

\newcommand{\half}{\frac{1}{2}}
\newcommand{\pd}{\frac{\pi}{2}}
\newcommand{\pq}{\frac{\pi}{4}}

\newcommand{\fp}{f^+(r,\nu)}
\newcommand{\fpo}{f_0^+(r,\nu)}

\newcommand{\fm}{f^-(r,\nu)}
\newcommand{\fmo}{f_0^-(r,\nu)}

\newcommand{\fpmo}{f_0^{\pm}(r,\nu)}

\newcommand{\fpm}{f^{\pm}(r,\nu)}

\newcommand{\tfp}{\tilde{f}^+(r,\nu)}
\newcommand{\tfm}{\tilde{f}^-(r,\nu)}

\newcommand{\reg}{\varphi(r,\nu)}
\newcommand{\rego}{\varphi_0(r,\nu)}

\newcommand{\treg}{\tilde{\varphi}(r,\nu)}

\newcommand{\al}{\alpha(\nu)}
\newcommand{\be}{\beta(\nu)}

\newcommand{\alo}{\alpha_0(\nu)}
\newcommand{\beo}{\beta_0(\nu)}

\newcommand{\tal}{\tilde{\alpha}(\nu)}
\newcommand{\tbe}{\tilde{\beta}(\nu)}

\newcommand{\ds}{\displaystyle}

\makeatletter
\newcommand{\Section}[1]{\section{#1} \setcounter{equation}{0}}

\@addtoreset{equation}{section}
\makeatother

\begin{document}

\title{Local inverse scattering at a fixed energy for radial Schr\"{o}dinger operators and localization of the Regge poles. }
\author{Thierry Daud\'e \footnote{Research supported by the French National Research Projects AARG, No. ANR-12-BS01-012-01, and Iproblems, No. ANR-13-JS01-0006} $^{\,1}$, Francois Nicoleau \footnote{Research supported by the French National Research Project NOSEVOL, No. ANR- 2011 BS0101901} $^{\,2}$\\[12pt]
 $^1$  \small D\'epartement de Math\'ematiques. UMR CNRS 8088. \\
\small Universit\'e de Cergy-Pontoise \\
\small 95302 Cergy-Pontoise, France  \\
\small thierry.daude@u-cergy.fr\\
$^2$  \small  Laboratoire de Math\'ematiques Jean Leray, UMR CNRS 6629 \\
      \small  2 Rue de la Houssini\`ere BP 92208 \\
      \small  F-44322 Nantes Cedex 03 \\
\small francois.nicoleau@math.univ-nantes.fr \\}





\maketitle


\begin{abstract}

We study inverse scattering problems at a fixed energy for radial Schr\"{o}dinger operators on $\R^n$, $n \geq 2$. First, we consider the class $\mathcal{A}$ of potentials $q(r)$ which can be extended analytically in $\Re z \geq 0$ such that $\mid q(z)\mid \leq C \ (1+ \mid z \mid )^{-\rho}$, $\rho > \frac{3}{2}$. If $q$ and $\tilde{q}$ are two such potentials  and if the corresponding phase shifts $\delta_l$ and $\tilde{\delta}_l$  are super-exponentially close, then $q=\tilde{q}$. Secondly,
we study the class of potentials $q(r)$ which can be split into $q(r)=q_1(r) + q_2(r)$ such that $q_1(r)$ has compact support and $q_2 (r) \in \mathcal{A}$.
If $q$ and $\tilde{q}$ are two such potentials, we show that for any fixed $a>0$,
${\ds{\delta_l - \tilde{\delta}_l  \ = \ o \left( \frac{1}{l^{n-3}} \ \left( {\frac{ae}{2l}}\right)^{2l}\right)}}$ when $l \rightarrow +\infty$ if and only if $q(r)=\tilde{q}(r)$ for almost all $r \geq a$. The proofs are close in spirit with the celebrated Borg-Marchenko uniqueness theorem, and rely heavily on the localization of the Regge poles that could be defined as the resonances in the complexified angular momentum plane. We show that for a non-zero super-exponentially decreasing potential, the number of Regge poles is always infinite and moreover, the Regge poles are not contained in any vertical strip in the right-half plane. For potentials with compact support, we are able to give explicitly their asymptotics. At last, for potentials which can be extended analytically in $\Re  z \geq 0$ with $\mid q(z)\mid \leq C \ (1+ \mid z \mid )^{-\rho}$, $\rho >1$ , we show that the Regge poles are confined in a vertical strip in the complex plane.


\vspace{2cm}




\end{abstract}

\vspace{1cm}

\tableofcontents
\newpage


\Section{Introduction.} \label{Intro}

In quantum scattering theory, given a pair of Hamiltonians $(-\Delta +V, -\Delta)$ on $L^2 ({\R^n})$, $n \geq 2$, one of the main object of interest is the scattering operator $S$. This scattering operator $S$ commutes with $-\Delta$ and consequently, it reduces to a multiplication by an operator-function $S(\lambda)$, called {\it{the scattering matrix}}, in the spectral representation of the Hamiltonian $-\Delta$.

\vspace{0.2cm}\noindent
The goal of this paper is to adress the following question : can we determine the potential $V$ from the knowledge of the scattering matrix $S(\lambda)$ at
a fixed energy $\lambda >0$ ?

\vspace{0.2cm}\noindent
For exponentially decreasing potential, (i.e  when the potential $V$ is a "very small" perturbation), we can answer positively to this question (see Novikov's papers \cite{No1}, \cite{No2}), but we emphasize that, in general, the answer is negative. For instance, in dimension $n=2$, Grinevich and Novikov \cite{GrNo} construct a family of real {\it{spherically symmetric}} potentials in the Schwartz space such that the associated scattering matrices are equal to the identity. Such potentials are called {\it{transparent potentials}}. Similarly, in the three dimensional case, Sabatier \cite{Sa} found a class of radial transparent potentials $q_a$, $a \in \R$, real for $r>0$, which are meromorphic in the complex plane cut along the negative real axis with the following asymptotics:
\begin{equation}\label{exemplesab}
q_a (r) = ar^{-\frac{3}{2}}\ \cos(2r -\pq) + O \left( r^{-2-\epsilon} \right) \ \ \ {\rm{when}} \ r \rightarrow + \infty.
\end{equation}

\vspace{0.2cm}\noindent
However, in dimension $n \geq 3$, if we assume that the potential $V$ has a regular behaviour at infinity, (i.e, $V$ is the sum of homogeneous terms at infinity), and if we know the scattering matrix at a fixed energy up to a smooth operator, we can reconstruct the asymptotics of the potential, (see \cite{JoSa}, \cite{WeYa}).

\vspace{0.2cm}\noindent
We emphasize, that in classical scattering theory, the situation is drastically different : for a spherically symmetric perturbation and for a fixed energy $\lambda$ large enough, the classical scattering matrix $S_{cl}(\lambda)$ determines the potential, (see for instance \cite{Fi}, \cite{Jo}).

\vspace{0.5cm}\noindent
In this paper, we study a quantum inverse scattering problem for the
Schrödinger equation on $\R^n$, $n \geq 2$,
\begin{equation}\label{sch}
-\Delta u + V(x) u = \lambda u,
\end{equation}
with a fixed energy $\lambda$. Without loss of generality, we fix $\lambda = 1$ throughout this paper. We assume that the potential $V(x)$ is spherically symmetric, i.e
\begin{equation} \label{symsph}
V(x) = q(r) \ ,\ r = \mid x \mid.
\end{equation}
It is well known that the Schrödinger equation (\ref{sch}) can be reduced to a countable family of radial equations, (see for instance \cite{RS3}); indeed, we write:
\begin{equation}
L^2(\R^n) = L^2(\R^+, r^{n-1}dr) \otimes L^2(\mathbb{S}^{n-1}, d\sigma),
\end{equation}
and for functions $u(x)= f(r)g(\omega)$, where $r = \mid x\mid>0$, $\omega = \frac{x}{r} \in \mathbb{S}^{n-1} $, one has:
\begin{equation}\label{decompradial}
(-\Delta + V) \ f(r)g(\omega) = \left(  - \frac{d^2}{dr^2} - \frac{n-1}{r} \frac{d}{dr} +q(r)
- \frac{1}{r^2} \ \Delta_{{\mathbb{S}}^{n-1}}   \right)\ f(r)g(\omega).
\end{equation}
The operator $\Delta_{{\mathbb{S}}^{n-1}}$ appearing in (\ref{decompradial}) is the Laplace Beltrami operator on the sphere $\mathbb{S}^{n-1}$ and has pure point spectrum. Its eigenvalues are given by $k_l= -l^2 -l(n-2)$, for $l \geq 0$. It follows that:
\begin{equation}\label{reduc1}
L^2(\R^+, r^{n-1}dr) \otimes L^2(\mathbb{S}^{n-1}, d\sigma) = \bigoplus_{l \geq 0}  L^2(\R^+, r^{n-1}dr) \otimes K_l,
\end{equation}
where $K_l$ is the eigenspace of $\Delta_{\mathbb{S}^{n-1}}$ associated with the eigenvalue $k_l$. The restriction of the Schrödinger operator on each subspace $L_l = L^2(\R^+, r^{n-1}dr) \otimes K_l$ is given by \begin{equation}
-\Delta +V \  _{\mid L_l} =  - \frac{d^2}{dr^2} - \frac{n-1}{r} \frac{d}{dr}   - \frac{k_l}{r^2}+q(r).
\end{equation}
Finally, if we define the unitary operator $U$,
\begin{eqnarray*}
U :  L^2(\R^+, r^{n-1} dr) & \rightarrow       & L^2(\R^+,  dr) \\
       f     \ \ \ \ \               & \mapsto       & r^{\frac{n-1}{2}} f(r).
\end{eqnarray*}
and setting $\nu(l)= l+ \frac{n-2}{2}$, we obtain immediately  a new family of radial Schrödinger equations which will be the main object
of this paper:
\begin{equation}\label{eqradiale}
U \ \left( - \frac{d^2}{dr^2} - \frac{n-1}{r} \frac{d}{dr} +q(r)  - \frac{k_l}{r^2} \right) \ U^{-1} =
- \frac{d^2}{dr^2} + \frac{\nu(l)^2-\frac{1}{4}}{r^2} +q(r).
\end{equation}

\vspace{0.2cm}
One assumes that the potential $q(r)$ is piecewise continuous on $\R^{+*}$ and satisfies the following conditions:
\begin{eqnarray*}
\hspace{1cm} &(H_1)& \hspace{1cm} \int_0^1 r^{1-2\epsilon} \ \mid q(r) \mid \ dr \ < \infty \ \ {\rm{for \ some}} \ \epsilon>0. \\
\hspace{1cm} &(H_2)& \hspace{1cm} \int_1^{+\infty} \mid q(r) \mid \ dr \ < \infty.
\end{eqnarray*}


\noindent
The hypothesis $(H_1)$ insures that the centrifugal singularity dominates near the origin, whereas the hypothesis $(H_2)$ means that, at large distances, the potential is short range. Under these assumptions, it is well-known (\cite{RS3}, Theorem XI.53) that for all $l \geq0$, there exists a unique solution $\varphi(r,l)$ which is $C^1$ and piecewise $C^2$ on $(0,+\infty)$ satisfying:
\begin{equation}\label{lequationradiale}
-u'' + \left( \frac{\nu(l)^2-\frac{1}{4}}{r^2} +q(r) \right) \ u \ = \ u,
\end{equation}
with the boundary condition at $r=0$,
\begin{equation}\label{conditionen0}
\varphi (r,l) \ \sim \ r^{\nu(l) +\half} \ \ ,\ \ r \rightarrow 0.
\end{equation}
Moreover, this solution, called {\it{the regular solution}}, has the asymptotic expansion at infinity:
\begin{equation}\label{conditioneninfini}
\varphi(r,l) \ \ \sim \ 2 \ c(l) \ \sin \left( r- (\nu(l) - \half) \ \pd  + \delta_l \right) \ \ ,\ \ r \rightarrow +\infty,
\end{equation}
where the constant $c(l)$ is the modulus of the Jost function $\beta(\nu(l))$, (see Section 2 for details).

\vspace{0.2cm}
The quantities $\delta_l$ are called the {\it{phase shifts}} and are physically measurable. The scattering amplitude $T(\lambda, \omega, \omega')$, that is the integral kernel of the operator $S(\lambda)-1$, can be expressed by the phase shifts. For example, for $n=3$, one has the following relation:
\begin{equation}\label{amplitude}
T(\lambda=1, \omega, \omega') = - \frac{1}{\pi^2} \ \sum_{l=0}^{+\infty} (2l+1) \frac{e^{2i\delta_l} -1}{2i} \ P_l(\cos \theta),
\end{equation}
where $\omega, \omega' \in \mathbb{S}^2$, $\cos \theta = \omega \cdot \omega'$ and $P_l(t)$ are the Legendre polynomials.

\vspace{0.5cm}\noindent
So, we can reformulate our inverse problem as :
\vspace{0.2cm}

\centerline{{\it{Is the knowledge of the phase shifts $\delta_l$ enough to determine the potential $q(r)$ ?}}}

\vspace{0.2cm}\noindent
Of course, as we have said before, the answer is negative in general since the potentials appearing in \cite{GrNo}, \cite{Sa} are spherically symmetric. Nevertheless, for potentials with compact support, Ramm has obtained in \cite{Ra} a stronger result:

\vspace{0.2cm}

\begin{thm}
\hfill\break
Let $q$ and $\tilde{q}$ be two potentials locally integrable with compact support. We denote  $\delta_l$, (resp. $\tilde{\delta}_l$) the corresponding phase shifts. Consider a subset $\mathcal{L}$ of $\N^*$ that satisfies the M\"untz condition ${\ds{\sum_{l \in \mathcal{L}} \frac{1}{l} = \infty}}$,
and assume that $\delta_l =\tilde{\delta}_l$ for all $l \in \mathcal{L}$. Then $q = \tilde{q}$ a.e.
\end{thm}

\vspace{0.2cm}
The proof of the previous result is based on an old idea due to Regge \cite{Re}. This approach is called the method of the {\it{Complex Angular Momentum}} (CAM): we allow the angular momentum $l \in \N$ to be a complex number $\nu \in \C$. In some cases, it is possible to extend the equality $\delta_l =\tilde{\delta}_l$ for all $l \in \N$ into the equality $\delta(\nu) =\tilde{\delta}(\nu)$ for all $\nu \in \C \setminus \{poles\}$. Indeed, for some particular classes of holomorphic functions $f$, $f(\nu)$ is uniquely determined by its values at all the integers $f(l)$ (Carlson's theorem \cite{Boa}), or only for $f(l)$ with $l \in \mathcal{L}$, (Nevanlinna's class \cite{Ra}). Then we can often use this new amount of information to get the equality between the potentials $q$ and $\tilde{q}$.

\vspace{0.2cm}
The CAM method was used previously in a long paper by Loeffel \cite{Lo}. In this paper, Loeffel studied in great details the properties of a  meromorphic function $\sigma(\nu)$ in the domain $\Re \nu >0$. This function $\sigma(\nu)$ is called the {\it{Regge interpolation}}, and for $\nu = \nu(l)=l+ \frac{n-2}{2}$, we have $\sigma(\nu(l)) = e^{2i \delta_l}$,
where $\delta_l$ are the phase shifts. In particular, he showed that, if the Regge interpolation
$\sigma(\nu)$ and $\tilde{\sigma}(\nu)$ corresponding to two suitable potentials $q$ and $\tilde{q}$,
satisfy $\sigma(\nu)= \tilde{\sigma}(\nu) $ for $\Re \nu >0$ where both are holomorphic, then $q = \tilde{q}$.
Therefore, all the problem consists in finding the classes of potentials $q$ such that the data $\delta_l$ determine uniquely $\sigma(\nu)$ for $\Re \nu >0$. For instance, this is the case for potentials with compact support (\cite{Lo}, Theorem 3), or for potentials which can be extended holomorphically in the domain $\Re z >0$ and with exponential decay (see also the paper of Martin and Targonski \cite{MaTa}). We emphasize that Ramm's result  for potential with compact support is actually a by-product of (\cite{Lo}, Theorem 3) by Loeffel. Indeed, the function $\xi(r,\nu)= \frac{\varphi'(r,\nu)}{\reg}$ used in the proof of \cite{Lo}, where $\reg$ is the regular solution,  satisfies $\nu^{-1} \xi(r,\nu)= O(1)$ for $\Re \nu$ large and $r>0$ fixed.  Hence, this last function lies in the Nevanlinna class for $\Re \nu$ large enough. Nevertheless, Ramm's proof has the advantage to be shorter.

\vspace{0.2cm}
In 2011, M. Horvath \cite{Ho1} also used the CAM approach and announced the following result in the three dimensional case: \par
\noindent
{\it{ Assume that the potential $q(r)$ satisfies
\begin{equation}\label{hypHorvath}
\int_0^1 r^{1-\epsilon} \ \mid q(r) \mid \ dr < \infty, \ \ {\rm{and}} \ \ \int_1^{+\infty} r \mid q(r) \mid \ dr < \infty,
\end{equation}
for some $\epsilon >0$. Then, the phase shifts $\delta_l$ for all $l \in \N$ determine $q(r)$ uniquely}}.

\vspace{0.2cm}
This result seems to us incorrect. M. Horvath claims that we can naturally extend the phase shifts $\delta_l$ as a holomorphic function $\delta(\nu), \ \Re \nu >0$ such that $\delta (l +\half)= \delta_l$. But, as we will see later in this paper, for any potentials with compact support (or even for potential super-exponentially decreasing) this function $\delta(\nu)$ has always an infinite number of poles $\nu_n$ with $\Re \nu_n \ \rightarrow + \infty$ as $n \rightarrow +\infty$. Hence, one of the goals of this paper is to give a correct answer to some inverse scattering problems for radial potentials. Nevertheless, we emphasize that all the asymptotics of the phase shifts $\delta_l$ as $l \rightarrow +\infty$,  obtained in (\cite{Ho1}, Corollary 1), are rigorously exact.

\vspace{0.2cm}
Another goal of this paper is to obtain a local uniqueness result from the data consisting of the phase shifts, close in spirit with the celebrated local Borg Marchenko's uniqueness theorem (\cite{Be}, \cite{GS}, \cite{Si}), and to precise an open question formulated by  Vasy-Wang \cite{VaWa}, (for a class of potentials not necessary with spherical symmetry):

\vspace{0.5cm}
\par
{\it{
Let us consider the pair of Hamiltonians $(-\Delta, -\Delta +V(x))$ on $L^2(\R^n)$. Assume that the potential $V$ is smooth, real, and can be split into  $V(x)=U(x)+W(x)$ such that $W(x)$ is exponentially decreasing and $U(x)$ is dilatable analytically \textit{i.e.}
\begin{equation}\label{dilatable}
  \exists C>0, \ \forall \theta \ {\rm{in \ a \ small \ complex \  neighbourhood\ of\ }} 0, \ \mid U(e^{\theta}x) \mid \leq C \ (1+\mid x\mid)^{-\rho}, \ \rho>1.
\end{equation}
Does the scattering matrix $S(\lambda)$ at a fixed energy $\lambda>0$ determine uniquely  $V(x)$ ?}}

\vspace{0.5cm}
As it was pointed to us by Roman Novikov, the family of central transparent potentials constructed in \cite{GrNo} contains a subset of analytic
potentials in $\R^2$, but it is not clear for us  that these potentials can be extended in a complex angular sector containing the positive axis. Thus, for generic dilatable analytically potentials, this question remains open. Nevertheless, we shall see in this paper, that for central potentials $V(x)= W(x)+U(x)$ where $W$ has a compact support and $U(x)$ can be extended to an holomorphic function in $\Re z \geq 0$, the answer is positive. Moreover, our result is {\it{local}} in nature. We think that our result is still true if $W(x)$ decays exponentially, but we did not succeed in proving it.

\vspace{0.5cm}\noindent
First of all, let us begin by a global uniqueness theorem for dilatable analytically  potentials.

\vspace{0.2cm}\noindent
{\it{Definition : We say that the potential $q(r)$ belongs to the class $\mathcal{A}$ if $q$ can be extended analytically
in $\Re z \geq 0$ and satisfies in this domain the estimate $\mid q(z) \mid  \leq C \  (1+ \mid z \mid)^{-\rho}, \ \rho > \frac{3}{2}$.}}

\vspace{0.2cm}\noindent
Our main first result is the following:

\begin{thm}\label{analyticalcase}
\hfill\break
Let $q(r)$ and $\tilde{q}(r)$ be two potentials  belonging to the class $\mathcal{A}$. We denote $\delta_l$, (resp. $\tilde{\delta}_l$) the corresponding
phase shifts. Assume that $\delta_l$ and  $\tilde{\delta}_l$ are super-exponentially close, i.e for all $A>0$,
\begin{equation}\label{hypodelta}
\delta_l - \tilde{\delta}_l \ = \ O (e^{-Al}), \ l \rightarrow + \infty.
\end{equation}
Then, $q(r) =\tilde{q}(r)$  on $(0, +\infty)$.
\end{thm}

\vspace{0.3cm}\noindent
\begin{rem}
\hfill\break
\begin{enumerate}
\item  Let us consider the potential ${\ds{ q(r) =\frac{e^{-ar}} {(r+1)^2}}}$ where $a>0$. Clearly, $q(r)$ satisfies the hypothesis of Theorem \ref{analyticalcase}. Thanks to the proof of Proposition \ref{super} in Section 6 (see also \cite{Ho1}, Corollary 1), one gets:
\begin{equation*}
\delta_l = O \left( \frac{1}{\sqrt{l}} \ e^{- \eta l } \right) \ ,\ l \rightarrow +\infty,
\end{equation*}
where $\cosh \eta = 1+\frac{a^2}{2}$. So, if $\delta_l$ and  $\tilde{\delta}_l$ are only exponentially close, we can not obtain a uniqueness result.
\item The family of the meromorphic transparent potentials given in \cite{Sa} satisfy  $q(r)=O(r^{-\frac{3}{2}})$ at infinity.
Unfortunately, we are not able to verify if these potentials satisfy all the hypotheses of Theorem \ref{analyticalcase}, but  if it is the case,
it would imply that the assumption $\rho > \frac{3}{2}$ is sharp.
\item  In the Born approximation (that is in the linear approximation near zero potential), it is well-known that the scattering amplitude can be approximated by:
\begin{equation}\label{Born}
T(\lambda, \omega, \omega') \thickapprox \hat{V}(\sqrt{\lambda}\omega- \sqrt{\lambda}\omega'),
\end{equation}
where $\hat{V}$ is the Fourier transform of the potential (non necessarily spherically symmetric). We emphasize that
Theorem \ref{analyticalcase} is coherent with the Born approximation. Indeed, it is not difficult to prove that, for a large set of potentials  in the class $\mathcal{A}$, the restriction on any ball of $\hat{V}(\xi)$ determines uniquely the potential, (see Theorem \ref{UniqueFourier} for details).

\end{enumerate}
\end{rem}

\vspace{0.5cm}\noindent
Now, let us define our class of potentials $\mathcal{C}$ which will be useful for our local inverse problem.

\vspace{0.5cm}\noindent
{\it{Definition :  We say that the potential $q(r)$ belongs to the class $\mathcal{C}$ if $q(r)$ can be split into $q(r)=q_1(r)+q_2(r)$ where $q_1$ is piecewise continuous on $(0,+\infty)$. Moreover,
 \begin{enumerate}
 \item $q_1$ has compact support and satisfies the hypothesis $(H_1)$.
 \item $q_2(r) \in \mathcal{A}$.
\end{enumerate}
 }}

\vspace{0.5cm}\noindent
Before giving our second main result, we recall the following fact. In the three dimensional case, let $q(r)$ be a piecewise continuous potential with support in $[0,a]$ and assume that $q(a-0) \not=0$. In (\cite{Ho1}, Corollary 1, Eq. (17) with $\nu = l+\half$), M. Horvath proved that:
\begin{equation}\label{equivalent}
\delta_l \ \sim \ - \frac{q(a-0)}{2} \ \left(\frac{a}{2l}\right)^3 \ \left( \frac{ae}{2l}\right)^{2l} \ ,\ l \rightarrow + \infty.
\end{equation}
We also refer the reader to \cite{Ra1} where the formula for the radius of the support of the potential is calculated from the scattering data.

\vspace{0.1cm}\noindent
Obviously, it follows that if $q$ belongs to the class $\mathcal{C}$ with $q_2=0$, the phase shifts $\delta_l$ must satisfy the same asymptotics as in (\ref{equivalent}). Having this result in mind, we can state our second main result by :

\vspace{0.3cm}\noindent

\begin{thm}\label{Mainresult}
\hfill\break
Let $q(r)$ and $\tilde{q}(r)$ be two potentials belonging to the class ${\mathcal{C}}$ and let $\delta_l$, (resp. $\tilde{\delta}_l$) the corresponding
phase shifts. Let us fix $a>0$. Then, the two following assertions are equivalent:
\begin{eqnarray*}
&(A_1)& \hspace{1cm} \delta_l - \tilde{\delta}_l \ = \ o \left( \frac{1}{l^{n-3}}\ \left( \frac{ae}{2l}\right)^{2l}\right) \ \ , \ \ l \rightarrow + \infty. \\
&(A_2)& \hspace{1cm}  q(r) =\tilde{q}(r) \ {\it{for \ almost \  all}}  \ r \in (a,+\infty).
\end{eqnarray*}
\end{thm}

\vspace{0.3cm}\noindent
\begin{rem}
\hfill\break
\begin{enumerate}
\item When $q$ and $\tilde{q}$ have compact support, Theorem \ref{Mainresult} was essentially proved by M. Horvath \cite{Ho} in the three dimensional case.
\item Actually, we have a better result for the implication $(A.2) \Longrightarrow (A.1)$. If $q(r)= \tilde{q}(r)$ almost everywhere on $(a,+\infty)$, we can prove (see Proposition \ref{necessarycondition}):
\begin{equation}
\delta_l - \tilde{\delta}_l \ = \ O \left( \frac{1}{l^{n}}\ \left( \frac{ae}{2l}\right)^{2l}\right) \ \ , \ \ l \rightarrow + \infty.
\end{equation}
\item We also think that Theorem \ref{Mainresult} remains true if, in the definition of the class ${\mathcal{C}}$, we can allow the potential $q_1$ to be exponentially decreasing. We were not able to prove it since Corollary \ref{compact} fails when the potential $q-\tilde{q}$ decays exponentially.
\item At last, if $q_1$, (resp. $\tilde{q}_1$) is super-exponentially decreasing (see Section 6 for the definition), following the proof of Theorem \ref{Mainresult}, one can get $q_2 (r)= \tilde{q}_2(r)$ for all $r>0$.
\end{enumerate}
\end{rem}

\vspace{0.8cm}\noindent
{\it{Outlines of the proof:}}

\vspace{0.2cm}
1 - The implication $(A_2) \Longrightarrow (A_1)$ in Theorem \ref{Mainresult} is easy to prove and can be found in \cite{Ho}, but for the convenience's reader, we shall give here a shorter proof. \\

2 - Theorem \ref{analyticalcase} and the implication $(A_1) \Longrightarrow (A_2)$ in Theorem \ref{Mainresult} follow in spirit the {\it{local Borg Marchenko's uniqueness theorem}}, (see  \cite{Be, GS, Si, Te}).

\vspace{0.1cm}\noindent
Let us explain briefly our approach:  we fix $r>0$ and  we define $F(r,\nu)$ as an application of the complex variable $\nu$ by:
\begin{equation}\label{definitionF0}
F(r,\nu) = \fp \tfm - \fm\tfp,
\end{equation}
where $f^{\pm}(r, \nu)$ and $\tilde{f}^{\pm}(r, \nu)$ are the Jost solutions associated with the potentials $q$ and $\tilde{q}$.

\vspace{0.2cm}\noindent
We are able to prove that this application is even with respect to $\nu$, holomorphic on the whole complex plane $\C$, and of order 1 with infinite type. Moreover, $F(r,\nu)$ is bounded on the imaginary axis $i\R$, and for all $r>0$ in Theorem \ref{analyticalcase}, (resp. for all $r\geq a$ in Theorem \ref{Mainresult}), we can show that $F(r, \nu) \rightarrow 0$ when $ \nu \rightarrow +\infty$, ($\nu$ real).

\vspace{0.2cm}\noindent
So, using the Phragmen-Lindelöf Theorem on each quadrant of the complex plane, we deduce that $F(r,\nu)$ is identically equal to zero, which implies easily the uniqueness of the potentials for $r>0$ in Theorem \ref{analyticalcase}, (resp. for $r  \geq a$ in Theorem \ref{Mainresult}).

\vspace{1cm} \noindent
This same approach has been used recently to study scattering inverse problems for asymptotically hyperbolic manifolds (see \cite{DGN,DKN,DN3}). In the hyperbolic setting, a Liouville transformation changes the angular momentum variable in a spectral variable, and we can see the Jost solutions as suitable perturbations of the modified Bessel functions $I_{\nu}(z)$. In the hyperbolic context, the variable $\nu$ is {\it{fixed}} and depends only on the geometry of the manifold, whereas the variable $z$ ranges over $\C$.

\vspace{0.2cm}\noindent
However, in the Euclidean setting of  this paper, the situation is drastically different; as in the hyperbolic case, the regular solution and the Jost solutions are close (in some sense) to the Bessel functions $J_{\nu}(r)$ or to the Hankel functions $H_{\nu}^{(j)}(r)$, but  the complex angular momentum $\nu$ can be as large as possible and the radial variable $r$ ranges over the {\it{non-compact}} set $(0,+\infty)$. In general, this kind of situation leads to very cumbersome and complicated calculations and one has to use the Langer uniform asymptotic formula for Bessel functions of large order  \cite{Lan}. We emphasize that this is not the case in this paper : we only use elementary properties for the Bessel functions.

\vspace{0.5cm} \noindent
The proofs of this paper rely heavily on the localization of the {\it{Regge poles}}. By definition, the Regge poles are the zeros of the Jost function $\be$ (see Section 2 for details) or equivalently, are the poles of the meromorphic continuation of the phase shifts $\delta(\nu)$ for $\nu \in \C$. Following Regge's theory, the positions of these poles determine power-law growth rates for the scattering amplitude. Moreover, the low-energy scattering is well studied using the Regge poles. In particular, they provide a rigorous definition of resonances.

\vspace{0.2cm}\noindent
In this paper, we prove the following theoretical result concerning the localization of the Regge poles. For non-zero super-exponentially decreasing potentials, we show that the number of the Regge poles are always infinite and there aren't bounded to the right in the first quadrant of the complex plane. Our result contradicts Theorem 5.2 in \cite{HBM}, which says that for an integrable potential $q(r)$ on $(0,+\infty)$, there are finitely many Regge poles in the right-half plane. As it was pointed to us by M. Marletta \cite{Mar}, the error in \cite{HBM} comes from Eq. (5.3) : the Green kernel $\Theta(r,s)$ appearing in the integral is not bounded with respect to the complex angular momentum.
We emphasize that our theoretical result is confirmed later, where for potentials with compact support, the precise asymptotics of the Regge poles are calculated.

\vspace{0.2cm}\noindent
For potentials which can be extended analytically in $\Re  z \geq 0$ with $\mid q(z)\mid \leq C \ (1+ \mid z \mid )^{-\rho}$, $\rho > 1$, the  situation is different. We generalize a result of Barut and Diley \cite{BaDi}, (see also \cite{Bes}), and we show that the Regge poles are confined in a vertical strip in the complex plane.


\section{Review of scattering theory for central potentials.}

In this Section, we recall (without proofs) some results obtained by Loeffel in \cite{Lo}.

\vspace{0.2cm}\noindent
Following Regge's idea, we  consider the radial Schrödinger equation on $(0,+\infty)$ at the fixed energy $\lambda=1$, where the angular
 momentum $\nu$ is now supposed to be a complex number and with $\Re \nu \geq 0$:
\begin{equation}\label{Schradiale}
-u'' +\left( \frac{\nu^2-\frac{1}{4}}{r^2} +q(r) \right) u =u.
\end{equation}
Of course, when $\nu = \nu(l)= l+\frac{n-2}{2}$, we recover the family of radial Schrödinger equations (\ref{eqradiale}) coming from the separation of variables.

\vspace{0.2cm}\noindent
First, we define {\it{the regular solution}} $\reg$ which is a solution of (\ref{Schradiale})  satisfying the boundary condition at $r=0$:
\begin{equation}\label{regsolution}
\reg \ \sim \ r^{\nu +\half} \ \ ,\ \ r \rightarrow 0.
\end{equation}
If the  potential $q(r)$ is piecewise continuous and satisfies the hypothesis $(H_1)$, Loeffel shows that, for $r>0$ fixed, the map $\nu \rightarrow \reg$, (resp. $\nu \rightarrow \varphi'(r,\nu)$)  is holomorphic in $\Re \nu \geq -\epsilon$, and we have:
\begin{equation}\label{conjregular}
\forall \ \Re \nu \geq -\epsilon,\ \overline{\reg}=\varphi(r, \bar{\nu}).
\end{equation}

\vspace{0.4cm}\noindent
Similarly, if the potential satisfies the hypothesis $(H_2)$, we can define {\it{the Jost solutions}} $f^{\pm}(r,\nu)$ as the unique solutions of (\ref{Schradiale}) satisfying the boundary condition at $r=+\infty$:
\begin{equation}\label{solJost}
f^{\pm}(r,\nu) \ \sim  \ e^{\pm i r}\ ,\ \ r \rightarrow +\infty.
\end{equation}
For $r>0$ fixed, the maps $\nu \rightarrow f^{\pm}(r,\nu)$, (resp. $\nu \rightarrow {f^{\pm}}'(r,\nu)$), are holomorphic on $\C$ and are even functions. Moreover,
\begin{equation}\label{conjJost}
\forall \nu \in \C\ \ ,\ \  \overline{\fp} = f^- (r, \bar{\nu}).
\end{equation}

\vspace{0.4cm}\noindent
The pair of the Jost solutions is a fundamental system of solutions (FSS) of (\ref{Schradiale}). Hence we can write
\begin{equation}\label{SFS}
\reg = \alpha(\nu) \fp + \beta(\nu) \fm,
\end{equation}
where $\alpha(\nu), \ \beta(\nu) \in \C$ are called {\it{the Jost functions}}. We recall that the Wronskian of two functions $u,v$ is given by $W(u,v)=uv'-u'v$. So, it follows immediately from (\ref{solJost}) that
\begin{equation}\label{wronskien}
W(\fp, \fm)=-2i.
\end{equation}
Hence, one has:
\begin{equation}\label{Jostfunctionalpha}
\alpha(\nu) = \frac{i}{2} \ W(\reg, \fm),
\end{equation}
\begin{equation}\label{Jostfunctionbeta}
\beta(\nu) = -\frac{i}{2} \ W(\reg, \fp).
\end{equation}
We can deduce that the Jost functions are holomorphic in $\Re \nu \geq -\epsilon$ and satisfy
\begin{equation}\label{conjJostfunction}
\overline{\alpha(\nu)} = \beta(\bar{\nu}).
\end{equation}
\vspace{0.2cm}
Now, let us give some elementary properties of the Jost functions (see section 4 for the details): the Jost function $\al$ does not vanish in the first  complex quadrant  $\Re \nu \geq 0$, $\Im \nu \geq 0$, whereas $\be$ does not  vanish in the fourth complex quadrant $\Re \nu \geq 0, \ \Im \nu \leq 0$.
The zeros of the Jost function $\be$ (belonging to the first quadrant), are called {\it{the Regge poles.}} These are the poles of the so-called {\it{Regge interpolation}}:
\begin{equation}\label{interpolation}
\sigma(\nu) = e^{i\pi(\nu +\half)} \ \frac{\al}{\be}.
\end{equation}
When $\nu >0$, it follows from (\ref{conjJostfunction}) that $\mid \sigma(\nu)\mid =1$, thus we can define {\it{the generalized phase shifts}} $\delta (\nu)$ as a continuous function in $(0,+\infty)$ through the relation:
\begin{equation}\label{phaseshifts}
\sigma (\nu) = e^{2i \delta(\nu)}.
\end{equation}
The generalized phase shifts become unique if we impose the condition $\delta(\nu) \rightarrow 0$ when $\nu \rightarrow +\infty$. Then, we deduce from (\ref{solJost}), (\ref{SFS}) and (\ref{phaseshifts}) that:
\begin{equation}\label{comportement}
\reg \ \sim \ 2 \ \mid \be \mid \ \sin\left( r-(\nu-\half)\pd + \delta(\nu) \right) \ \ ,\ \ r \rightarrow + \infty.
\end{equation}
In particular,  when $\nu = \nu(l)$,  (\ref{conditioneninfini}) implies that the generalized phase shifts are related to the {\it{physical}} phase shifts:
\begin{equation}\label {vraiphase}
\delta(\nu(l)) = \delta_l \ \ ,\ \ l \in \N.
\end{equation}

\vspace{0.2cm}\noindent
We emphasize that, if we can show  that for suitable potentials $q(r)$, there are no Regge poles in  a simply connected domain, the equation (\ref{phaseshifts}) allows us to define $\delta(\nu)$ as a holomorphic function in this domain.

\vspace{1cm}
Now, let us examine the free case, when the potential $q(r)=0$. In this case, the equation (\ref{Schradiale}) is a standard Bessel equation (see \cite{Leb}, p. 106). As a consequence,  we have explicit formulae for the regular solution $\rego$, the Jost solutions
$f_0^{\pm}(r,\nu)$ and the Jost functions $\alo$, $\beo$. First, we denote:
\begin{equation}\label{Anu}
A(\nu) = \sqrt{\frac{2}{\pi}} \ 2^{\nu} \ \Gamma (\nu+1).
\end{equation}
Then, we have:
\begin{equation}\label{phi0}
\rego = A(\nu) \ \sqrt{\frac{\pi r}{2}}\ J_{\nu}(r),
\end{equation}
\begin{equation}\label{fp0}
\fpo = e^{i(\nu+\half) \pd} \ \sqrt{\frac{\pi r}{2}}\ H_{\nu}^{(1)}(r),
\end{equation}
\begin{equation}\label{fm0}
\fmo = e^{-i(\nu+\half) \pd} \ \sqrt{\frac{\pi r}{2}}\ H_{\nu}^{(2)}(r),
\end{equation}
where $J_{\nu}(r)$ is the Bessel function of order $\nu$ and $H_{\nu}^{(j)}(r)$ are the Hankel functions of order $\nu$, (see the Appendix for details). So using
(\ref{Jostfunctionalpha}), (\ref{Jostfunctionbeta}) and (\cite{Leb}, Eq. (5.9.3)), we obtain for $\Re \nu \geq -\epsilon$,
\begin{equation}\label{alphazero}
\alo = \half \ A(\nu) \ e^{-i (\nu +\half) \pd},
\end{equation}
\begin{equation}\label{betazero}
\beo = \half \ A(\nu) \ e^{i (\nu +\half) \pd}.
\end{equation}
Obviously, it follows from (\ref{interpolation}), (\ref {phaseshifts}),  (\ref{alphazero}) and (\ref{betazero}) that in the free case the Regge interpolation
$\sigma_0(\nu) =1$ and the generalized phase shifts $\delta_0 (\nu) =0$.

\section{The regular solution $\reg$.}

In this section, first, we recall very briefly  the results obtained by Loeffel \cite{Lo}. Secondly, we give a new integral representation for the regular solution $\reg$.

\vspace{0.5cm}\noindent
Using the method of variation of constants, it is easy to see (\cite{AlRe}, Eq. $(3.14)$, or \cite{Lo}) that for $\Re \nu \geq -\epsilon$,

\begin{equation} \label{repint1}
\reg = r^{\nu + \half} + \frac{1}{2\nu} \int_0^r \left( (\frac{s}{r})^{\nu} - (\frac{r}{s})^{\nu} \right) \sqrt{rs} \ (1-q(s)) \ \varphi(s,\nu) \ ds.
\end{equation}
It follows that for a fixed $r>0$, $\nu \rightarrow \reg$, (resp. $\nu \rightarrow \varphi'(r,\nu)$) are holomorphic for $\Re \nu \geq -\epsilon$ and one has the following estimate (\cite{Lo}, Eq. (3))
\begin{equation}\label{estreg}
\mid \reg \mid \ \leq \ r^{\Re \nu + \half} \ exp \ \left( \frac{ r^{2\epsilon}}{\mid \nu \mid } \  \int_0^r s^{1-2\epsilon} \ \mid q(s)-1 \mid \ ds \right) \ \ {\rm{for}} \ \mid \nu \mid \geq 2 \epsilon.
\end{equation}
Moreover, for $\Re \nu \geq -\epsilon$, one has:
\begin{equation}\label{conjreg}
\overline{\reg} = \varphi (r, \bar{\nu}).
\end{equation}

\vspace{1cm}
We prefer to work with another Green kernel in order to obtain better estimates for the regular solution with respect to $\nu$.
We can find the next Proposition \ref{newrep} implicitly in (\cite{Ho1}, Eq. $(35)$), but in this paper, M. Horv\' ath wrote this lemma using the generalized phase shifts $\delta(\nu)$ for $\Re \nu >0$ which are, according to us, not well defined, in the presence of Regge poles.

\vspace{0.3cm}\noindent
First, let us introduce the Green kernel $K(r,s,\nu)$ we shall use. We denote:
\begin{equation}\label{u}
u(r) = \sqrt{\frac{\pi r}{2}} \ J_{\nu} (r),
\end{equation}
\begin{equation}\label{v}
v(r) = -i \sqrt{\frac{\pi r}{2}} \ H_{\nu}^{(1)} (r).
\end{equation}
The pair ($u(r), \ v(r))$ is a (FSS) of the equation (\ref{Schradiale}) when $q(r)=0$. The wronskian $W(u,v)=1$. Moreover, we have the elementary following lemma, (see \cite{Leb} for details):

\begin{lemma}\label{estimations}
\hfil\break
(i) When $r \rightarrow 0$,
\begin{eqnarray}\label{estimater=0}
u(r) &\sim &  \frac{1}{A(\nu)} \ r^{\nu+\half} ,\\
v(r) &\sim& - \frac{1}{2\nu} \  A(\nu) \ r^{-\nu +\half} \ {\rm{if}} \ \nu \not=0\ \ .
\end{eqnarray}
(ii) When $r \rightarrow +\infty$,
\begin{equation}\label{estimater=infini}
u(r) \sim \sin \left( r-(\nu-\half) \pd \right) \ \ ,\ \
v(r) \sim - e^{i\left( r-(\nu -\half)\pd \right) }.
\end{equation}
\end{lemma}

\vspace{0.5cm}\noindent
We define the Green kernel $K(r,s,\nu)$ for $r,s >0$ and $\Re \nu \geq 0$ by:

\begin{equation}\label{greenkernel}
K(r,s , \nu) = u(s)\ v(r) \ {\rm{if}} \ s \leq r \ \ ,\ \ K(r,s , \nu) = u(r)\ v(s) \ {\rm{if}} \ s \geq r .
\end{equation}

\vspace{0.2cm}\noindent
The following elementary Proposition will be powerful to prove our local uniqueness result:

\vspace{0.2cm}

\begin{prop} \label{newrep}
\hfil\break
Let $q(r)$ be a potential satisfying $(H_1)$ and $(H_2)$. For $\Re \nu > 0$, one has:
\begin{equation}\label{newrep1}
\reg = -2i\be \ e^{-i \pd (\nu-\half)} \ u(r) + \int_0^{+\infty} K(r,s, \nu) q(s)  \varphi(s,\nu) \ ds.
\end{equation}
\begin{equation}\label{newrep2}
\al \ e^{i \pd (\nu-\half)} + \be \ e^{-i \pd (\nu-\half)} \ =\ -\int_0^{+\infty} u(s) q(s) \varphi(s,\nu) \ ds.
\end{equation}
\end{prop}

\begin{proof}
Using the method of variation of constants, there exists $A,B \in \C$ such that
\begin{equation}\label{integrale}
\reg = A u(r) + B v(r) + \int_0^{+\infty} K(r,s, \nu) q(s)  \varphi(s,\nu) \ ds.
\end{equation}
Note that the (RHS) of (\ref{integrale}) is well defined with the help of Lemma \ref{estimations}. Then, we write
\begin{equation}\label{int1}
\reg = \left( A+\int_r^{+\infty} v(s)q(s)\varphi(s,\nu) \ ds \right) \ u(r) + \left( B+\int_0^r u(s)q(s)\varphi(s,\nu) \ ds \right) \ v(r).
\end{equation}
It follows from Lemma \ref{estimations} that
\begin{equation}
\reg = (B+o(1)) \ v(r) \ ,\ r \rightarrow 0.
\end{equation}
Since we have $\reg \sim r^{\nu +\half}$ when $r \rightarrow 0$, we have to take $B=0$. Then, the equation (\ref{int1}) implies, when $r \rightarrow +\infty$:
\begin{eqnarray*}
\reg &\sim & A u(r) + \left( \int_0^{+\infty} u(s)q(s) \varphi(s,\nu) \ ds  \right) \ v(r)\ , \\
     &\sim & A \sin \left( r-(\nu-\half) \pd \right) - \left( \int_0^{+\infty} u(s)q(s) \varphi(s,\nu) \ ds \right) \  e^{i\left( r-(\nu -\half)\pd \right) }.
\end{eqnarray*}
In other way, from (\ref{solJost}) and (\ref{SFS}), we deduce
\begin{equation}
\reg \sim \al e^{ir} + \be e^{-ir} \ \ ,\ \ r \rightarrow + \infty.
\end{equation}
So, we obtain easily
\begin{eqnarray}
\al &=& \left( \frac{A}{2i} - \int_0^{+\infty} u(s) q(s) \varphi (s,\nu) \ ds \right) \ e^{-i (\nu -\half)\pd  },\\
\be &=& - \frac{A}{2i} \ e^{i (\nu -\half)\pd  },
\end{eqnarray}
which implies the lemma.
\end{proof}

\section{The Jost solutions $f^{\pm} (r,\nu)$.}

As for the regular solution, using the method of variation of constants, the Jost solutions for $r>0$ are given by, (\cite{Lo}, Lemma 6):
\begin{equation}\label{integralequationjost}
f^{\pm} (r, \nu) = e^{\pm ir} + \int_r^{+\infty} \sin (r-s) \left( \frac{\nu^2-\frac{1}{4}}{s^2} +q(s) \right) \ f^{\pm}(s,\nu) \ ds.
\end{equation}
We can deduce easily that the maps $ \nu \rightarrow f^{\pm}(r,\nu)$, (resp. $ \nu \rightarrow {f^{\pm}}'(r,\nu)$), are holomorphic on $\C$, are even functions and
satisfy the following estimate:
\begin{equation}\label{estifpm}
\mid f^{\pm}(r,\nu) \mid \ \leq \ exp\ \left(\int_r^{+\infty} \mid  \frac{\nu^2-\frac{1}{4}}{s^2} +q(s) \mid \ ds \right) .
\end{equation}
In particular, for a fixed $r>0$, $ \nu \rightarrow f^{\pm}(r,\nu)$ is an entire function of order $2$, i.e there exists $A,B >0$ such that
$\mid f^{\pm}(r,\nu) \mid \ \leq \ A \ e^{\ B \ \mid \nu \mid^2}, \ \forall \nu \in \C$.

\vspace{1cm}
In this section, we  shall obtain new useful estimates for the Jost solutions $f^{\pm}(r,\nu)$ and we shall see that in particular the Jost solutions are entire
functions of order $1$ with infinite type with respect to  $\nu \in \C$, (see below for the definition). Before studying $f^{\pm}(r,\nu)$, let us examine in detail the free case. First, let us recall some well-known definitions for holomorphic functions (see for instance \cite{Lev}).

\vspace{0.5cm}\noindent
{\it{Definitions}}
\hfill \break
Let $ f:\C \rightarrow \C$ an entire function of the complex variable $z$. Let
\begin{equation}
M(r) = \sup_{\mid z\mid =r} \ \mid f(z) \mid.
\end{equation}
We say that $f$ is of order $\rho$ if
\begin{equation}
\limsup_{r \rightarrow +\infty} \ \frac {\log \log M(r)}{\log r} \ =\ \rho.
\end{equation}
A function of order $\rho$ is said of type $\tau$ if
\begin{equation}
\limsup_{r \rightarrow +\infty} \ \frac {\log  M(r)}{\log r^{\rho}} \ =\ \tau.
\end{equation}
If  $\tau = +\infty$, we say that the function $f$ is of order $\rho$ with infinite (or maximal) type.

\vspace{0.2cm}
\begin{lemma}\label{orderonezero1}
\hfill\break
For $r>0$ fixed, the free Jost solutions $\fpmo$ are holomorphic functions of order $1$ with infinite type with respect to $\nu$.
\end{lemma}

\begin{proof}
For instance, let us examine $\fpo$. We recall that $\fpo$ is even with respect to $\nu$, so it suffices to estimate $\mid \fpo \mid $ for $\Re \nu \geq 0$.
Using (\ref{fp0}), we have:
\begin{equation}
\mid \fpo \mid = e^{-\Im \nu \ \pd} \ \sqrt {\frac{\pi r}{2}} \ \mid H_{\nu}^{(1)} (r)\mid .
\end{equation}
On one hand, Theorem \ref{estproduitbessel} implies that, for $\Re \nu \geq 0$ and some $\delta \in (0,1)$ fixed:
\begin{equation}
\mid J_{\nu }(r)\  H_{\nu}^{(1)} (r) \mid \ \leq \ C_{\delta} \ e^{\pi \mid \Im\nu\mid}\ (1 + \Re \nu)^{-\frac{\delta}{2}} \ r^{\frac{\delta -1}{2}}.
\end{equation}
On the other hand, using Proposition \ref{unifBessel}, we have for $\Re \nu \geq 0$,
\begin{equation}
J_{\nu }(r) =  \  \frac{1}{ \Gamma(\nu+1)} \ \left(\frac{r}{2}\right)^{\nu} \ (1+o(1)) \ \ ,\  \ \nu \rightarrow \infty.
\end{equation}
Hence, it follows that there exists suitable  constants $A,B>0$ such that, for all $\Re \nu \geq 0$,
\begin{equation}\label{fpotypeinfini}
\mid \fpo \mid  \leq  A \ e^{B\  \mid \nu \mid} \  \mid \Gamma(\nu+1) \mid.
\end{equation}
Thus, the lemma follows from Stirling's formula, (\cite{Leb}, Eq. (1.4.24)).
\end{proof}

\vspace{0.2cm}
In the next Lemma, we precise the localization of the zeros of the free Jost solutions $\fpo$, $r>0$ fixed, with respect to the complex variable $\nu \in \C$. Of course, these zeros are also the zeros of the Hankel functions $H_{\nu}^{(j)} (r)$ as a function of its order $\nu$. We emphasize that this result is actually  a general fact for the Jost solutions $\fpm$. The proof is inspired from \cite{MaKo}. We recall that the first open quadrant of the complex plane is the set of the complex number $\nu$ such that $\Re \nu > 0$ and $\Im \nu > 0$, the second open quadrant is the set of the complex number $\nu$ such that $\Re \nu < 0$ and $\Im \nu > 0$, ...

\begin{lemma}\label{nonzero}
\hfill\break
For $r>0$ fixed, the zeros of the  free Jost solution $\fpo$, (resp. $\fmo$) as function of  $\nu$  belong to the first and third open quadrant, (resp. the second and the fourth open quadrant).
\end{lemma}

\begin{proof}
For instance, assume that $f_0^+(r_0, \nu)=0$ for some $r_0>0$ fixed. Using (\ref{Schradiale}) with $q=0$, we see that
\begin{equation}\label{deriveeW}
\frac{d}{dr}\ \left[W(\fpo, \overline{\fpo})\right] = \frac{\bar{\nu}^2 -\nu^2}{r^2} \mid \fpo\mid^2.
\end{equation}
Integrating (\ref{deriveeW}) over $[r_0,+\infty[$ and using (\ref{solJost}), we obtain:
\begin{equation}
\int_{r_0}^{+\infty} \frac{\bar{\nu}^2 -\nu^2}{r^2} \mid \fpo\mid^2 = -2i,
\end{equation}
or equivalently
\begin{equation}
2 \Re \nu \ \Im \nu \ \int_{r_0}^{+\infty} \frac{\mid \fpo\mid^2}{r^2}= 1.
\end{equation}
It follows that $\Re \nu \ \Im \nu >0$ which implies the Lemma.
\end{proof}

\vspace{1cm}
Now, let us study the Jost solutions $\fpm$. As in the previous section, we shall establish a new integral representation which will be useful for our inverse problem. Instead of using the Green kernel $K(r,s,\nu)$, we prefer to use a more convenience one, $N(r,s,\nu)$. Using the notation (\ref{u}), (\ref{v}), we set:
\begin{equation}\label{greenN}
N(r,s,\nu) = u(r) v(s)-u(s)v(r).
\end{equation}
Clearly, we have a new integral representation:
\begin{equation}\label{Jostint}
f^{\pm} (r,\nu) = f_0^{\pm}(r,\nu) + \int_r^{+\infty} N(r,s,\nu) q(s) f^{\pm}(s,\nu) \ ds.
\end{equation}

\noindent
As a by-product of this integral representation,  we can deduce that the Jost solutions, for $r>0$ fixed, are bounded when $\nu$ belongs to the
imaginary axis $i\R$ :

\begin{prop}\label{fpmimaginary}
\hfil\break
Assume that the potential $q$ satisfies the hypothesis $(H_2)$. Then, for all $r>0$ and $y \in \R$,
\begin{equation}\label{bornesuriR}
\mid f^{\pm} (r,iy) \mid \ \leq \  2^{\frac{1}{4}}\ exp\ \left( \sqrt{2} \int_r^{+\infty} \mid q(s) \mid \ ds\right).
\end{equation}
\end{prop}

\vspace{0.2cm}
\begin{proof}
We consider the case $(+)$ only and we solve (\ref{Jostint}) by iterations. We set:
\begin{eqnarray*}
\psi_0(r) &=& f_0^+(r,iy), \\
\psi_{k+1}(r)&=&\int_r^{+\infty} N(r,s, iy) q(s) \psi_k(s) \ ds.
\end{eqnarray*}
\vspace{0.2cm}\noindent
For $y \in \R$, one has $\mid f_0^+(r, iy) \mid = e^{-\pd y} \sqrt{\frac{\pi r}{2}} \mid H_{iy}^{(1)}(r) \mid$. So, using the first estimate in Proposition \ref{imaxis}, one sees that:
\begin{equation} \label{estifp}
\mid f_0^+(r,iy) \mid \leq 2^{\frac{1}{4}}
\end{equation}
Moreover, recalling that  ${\ds{J_{\nu}(r) = \half \left( H_{\nu}^{(1)}(r) + H_{\nu}^{(2)}(r)\right)}}$, (see the Appendix, Section A1), we obtain immediately
\begin{equation}\label{autre}
N(r,s,\nu) = i \pq \sqrt{rs} \ \left( H_{\nu}^{(1)}(r)H_{\nu}^{(2)}(s)-H_{\nu}^{(1)}(s)H_{\nu}^{(2)}(r) \right).
\end{equation}
Hence, it follows from (\ref{autre}) and Proposition \ref{imaxis} again, that $\forall s \geq r >0, \ \forall y \in \R$,
\begin{equation}
\mid  N(r,s,iy) \mid \ \leq \ \sqrt{2}.
\end{equation}
Then, for all $k \in \N$,  we easily prove by induction:
\begin{equation}
\mid \psi_k (r) \mid \ \leq \ 2^{\frac{1}{4}} \ \frac{1}{k!} \ \left( \sqrt{2} \int_r^{+\infty} \mid q(s) \mid \ ds \right)^k.
\end{equation}
Since ${\ds{f^+(r,iy) = \sum_{k=0}^{+\infty} \psi_k (r)}}$, one obtains the Lemma.
\end{proof}

\vspace{0.5cm}\noindent
When the potential $q(r)$ decays faster at infinity, roughly speaking when $q(r) = O(r^{-\rho})$ with $\rho > \frac{3}{2}$, we are able to prove that, for $r>0$ fixed, the Jost solutions $\fp \sim \fpo$ for $\nu \rightarrow \infty$ in the second or fourth quadrant, and $\fm \sim \fmo$ for $\nu \rightarrow \infty$ in the first or third quadrant. Let us explain briefly our strategy, (for instance, let us the study $\fp$ with $\nu$ in the fourth quadrant):

\vspace{0.4cm}
\noindent
It follows from Lemma \ref{nonzero} that $\fpo$ does not vanish in the fourth quadrant. Then, for $0<r\leq s$, we can set:
\begin{equation}\label{defM}
M(r,s,\nu) =\frac{f_0^+(s,\nu)}{\fpo} \ N(r,s,\nu).
\end{equation}
Thus, setting ${\ds{g(r) = \frac{\fp}{\fpo}}}$, we obtain immediately from (\ref{Jostint}):
\begin{equation}\label{intg}
g(r)= 1 + \int_r^{+\infty} M(r,s,\nu) \ q(s) g(s) \ ds.
\end{equation}


\vspace{0.2cm} \noindent
So, in order to solve the integral equation (\ref{intg}), we need uniform estimates for the Green kernel $M(r,s,\nu)$ for $0<r\leq s$ and $\nu$ in the fourth quadrant. This is the goal of the next Lemma which is rather technical: indeed, since the potential $q(r)$ may have a singularity at $r=0$ and decays like $r^{-\rho}$ with $\rho > \frac{3}{2}$, we have to distinguish the two different regimes $s \leq 1$ and $s \geq 1$:

\begin{lemma}\label{estimationM}
\hfil\break
For all $\delta \in (0,1)$, there exists $C_{\delta}>0$ such that for all $\nu$ in the fourth quadrant with $|\nu| \geq 1$ and for all $0<r\leq s$,
\begin{equation*}
\mid M(r,s,\nu) \mid \leq C_{\delta}\ \mid \nu \mid^{-\frac{\delta}{2}} \ \min \ (s, s^{\frac{\delta+1}{2}}).
\end{equation*}
\end{lemma}

\begin{proof}
It follows from (\ref{fp0}), (\ref{fm0}) and (\ref{autre}) that
\begin{equation} \label{defM1}
M(r,s,\nu)= \frac{i}{2} \ \frac{f_0^+(s,\nu)}{\fpo} \ \left( \fpo f_0^-(s,\nu)-f_0^+(s,\nu) \fmo \right).
\end{equation}
By Lemma \ref{orderonezero1}, $\fpmo$ are of order $1$ with infinite type and by Lemma \ref{nonzero}, $\fpo$ does not vanish in the fourth quadrant. So using Theorem 12, p.22 in \cite{Lev1}  and its corollary p.24, we deduce that $\nu \rightarrow M(r,s,\nu)$ is (at most) of order $1$ with infinite type in the fourth quadrant.
\vspace{0.2cm}\par \noindent
So, roughly speaking, by the Phragmen-Lindel\'{o}f's theorem, it suffices to estimate $M(r,s,\nu)$ for $\nu \geq 0$ and $\nu=iy$ with $y\leq0$.

\vspace{0.5cm}
\noindent
{\it{1 - First, let us estimate $M(r,s,\nu)$ for $\nu > 0$ and $\nu =iy, \ y <0$ when $s\leq 1$.}}

\vspace{0.2cm} \noindent
This case is rather simple since the variables $r$ and $s$ belong to a compact set. For $r \leq s$, we write:
\begin{eqnarray*}
N(r,s,\nu) &=&  v(r)v(s) \ \left( \frac{u(r)}{v(r)}-\frac{u(s)}{v(s)} \right) = v(r)v(s) \ \int_s^r \left( \frac{u}{v} \right) '(t) \ dt ,\\
           &=& v(r) v(s) \ \int_s^r \frac{u'(t)v(t) -u(t)v'(t)}{v^2(t)}  \ dt = v(r) v(s) \ \int_r^s \frac{1}{v^2(t)}  \ dt,
\end{eqnarray*}
where we have used $W(u,v)=1$. So, recalling that ${\ds{\fpo = i e^{i(\nu+\half)\pd} v(r)}}$, we have:
\begin{equation}\label{kernelM}
M(r,s,\nu) = v^2(s) \int_r^s \frac{1}{v^2(t)}  \ dt.
\end{equation}
Now, from Proposition \ref{unifBessel}, we see that for $r$ in a compact set, one has the uniform asymptotics:
\begin{equation}\label{uniform11}
 v(r) = - \frac{1}{\sqrt{\pi}}  \ \Gamma (\nu) \ \left(\frac{r}{2}\right)^{-\nu+\half } \ (1+O(\frac{1}{\nu})) \ ,\ \ \nu \rightarrow + \infty.
\end{equation}
\noindent
We deduce that the Green $M(r,s,\nu)$ satisfies for $\nu > 0$ and $s \leq 1$, the following uniform estimate:
\begin{equation}
\mid M(r,s,\nu) \mid \ \leq \ C \ \  s^{-2\nu + 1} \ \int_r^s \ t^{2 \nu -1} \ dt \ \leq \  \frac{Cs}{ \nu},
\end{equation}
where $C$  does not depend on $r,s$ and $\nu$. In the same way, using Proposition \ref{unifBessel} again and
\begin{equation}
H_{iy}^{(1)} (r) = \frac { e^{\pi y} J_{iy}(r) - J_{-iy}(r)}{\sinh (\pi y)},
\end{equation}
we obtain:
\begin{equation}
v(r) =    \frac{i\sqrt{\pi}}{ \Gamma (1+iy) \sinh(\pi y)} \ \left(\frac{r}{2}\right)^{\half +iy} \ (1+O(\frac{1}{y})) \ ,\ \ y \rightarrow - \infty.
\end{equation}
So, as in the case $\nu >0$, we have the uniform estimate for $y < 0$ and $s \leq 1$:
\begin{equation}
\mid M(r,s,iy) \mid \ \leq \ \frac{Cs}{\mid y\mid}.
\end{equation}
Now, we set:
\begin{equation}
f(\nu):= \nu \ M(r,s,\nu).
\end{equation}
Clearly, $f$ is of order $1$ with infinite type in the fourth quadrant, $f$ is bounded on his boundary by $Cs$. So, using the Phragmen-Lindel\'{o}f's theorem, we obtain for all $\nu$ in the fourth quadrant,
\begin{equation}
\mid f(\nu) \mid \leq \ C s,
\end{equation}
which of course, implies the Lemma in the case $s \leq 1$ and for $\mid \nu \mid \geq 1$.

\vspace{0.5cm}
\noindent
{\it{2 - Now, let us estimate $M(r,s,\nu)$ for $\nu \geq 0$ and $\nu =iy, \ y\leq0$ when $s\geq 1$.}}

\vspace{0.2cm} \noindent
Instead of using the integral representation (\ref{kernelM}), we start from
\begin{equation}
M(r,s,\nu) = \frac{v(s)}{v(r)}\ N(r,s,\nu) = \frac{v(s)}{v(r)} \left( u(r)v(s)-u(s)v(r) \right),
\end{equation}
so
\begin{equation}
\mid M(r,s,\nu)\mid  \leq  \frac{|v(s)|}{|v(r)|}\  |u(r)v(s)| + |u(s)v(s)|.
\end{equation}
Now, we use the following trick : for $\nu \geq 0$, we have (see the Appendix, section A.1):
\begin{eqnarray}\label{trick}
\mid H_{\nu}^{(1)} (r) \mid^2 &=& H_{\nu}^{(1)} (r) \overline{H_{\nu}^{(1)} (r)} = H_{\nu}^{(1)} (r) H_{\nu}^{(2)} (r) \nonumber \\
                               &=& (J_{\nu}(r) + i Y_{\nu}(r)) \ (J_{\nu}(r) -i Y_{\nu}(r)) = J_{\nu}^2 (r) + Y_{\nu}^2(r).
\end{eqnarray}
Then,
\begin{equation}\label{trick1}
\mid v(r) \mid^2 = \frac{\pi r}{2}\  \left (J_{\nu}^2 (r) + Y_{\nu}^2(r) \right).
\end{equation}
Thus, using (\cite{Wa}, p. 446), we see that for $\nu > \half$, the application  $r \rightarrow \mid v(r)\mid$ is strictly decreasing on $(0,+\infty)$, whereas for $\nu \in [0, \half]$ this application is increasing.
\footnote{ The function we denote $Y_{\nu}(r)$ is sometimes denoted by $N_{\nu}(r)$ in the litterature on Bessel functions.}

\vspace{0.1cm}
\par\noindent
For $\nu> \half$, it follows from the previous remark and Corollary \ref{estGreenK}, that for any $\delta \in (0.1)$, there exists $C_{\delta}>0$ such that
\begin{eqnarray*}
\mid M(r,s,\nu)\mid  &\leq &   |u(r)v(s)| + |u(s)v(s)| \\
 &\leq& C_{\delta} \ (1+\nu)^{-\frac{\delta}{2}}\ \left( (rs)^{\frac{1+\delta}{4}} + s^{\frac{1+\delta}{2}} \right) \\
 &\leq& 2C_{\delta} \ (1+\nu)^{-\frac{\delta}{2}}\ s^{\frac{1+\delta}{2}} .
\end{eqnarray*}

\vspace{0.2cm}\noindent
Now, let us study the case $\nu \in [0,\half]$. By Lemma \ref{estimations}, $\mid v(r) \mid \rightarrow 1$ when $r \rightarrow + \infty$. So, as for such $\nu$, $r \rightarrow \mid v(r) \mid$ is an increasing function, one has $ \mid v(r) \mid \leq 1$ for all $r >0$. Then, one has :

\begin{eqnarray*}
\mid M(r,s,\nu)\mid  &\leq &   \frac{\mid u(r) \mid }{\mid v(r) \mid } \mid v(s)\mid^2+ |u(s)v(s)| \\
&\leq &   \frac{\mid u(r) \mid }{\mid v(r) \mid } + |u(s)v(s)|.
\end{eqnarray*}
It follows from (\ref{trick}) that:
\begin{equation}
\frac{\mid u(r) \mid^2 }{\mid v(r) \mid^2 }  \leq  \frac{J_{\nu}^2(r)}{J_{\nu}^2 (r) + Y_{\nu}^2(r)} \leq 1,
\end{equation}
since for $\nu$ real, $J_{\nu}(r)$ and $Y_{\nu}(r)$ are real, (see Appendix, section A.1). Thus, as previously, for $s \geq 1$ and $\nu \in [0,\half]$, we have:
\begin{equation}
\mid M(r,s,\nu)\mid  \leq    1 + |u(s)v(s)| \leq C_{\delta} \ s^{\frac{1+\delta}{2}}.
\end{equation}

\vspace{0.2cm}\noindent
It remains to study the case $\nu=iy$ with $y \leq 0$. It follows from (\ref{defM}) that
\begin{equation}
M(r,s,iy) = i \frac{\pi}{4} s \left( H_{iy}^{(1)}(s) H_{iy}^{(2)}(s) - (H_{iy}^{(1)}(s))^2 \ \frac{H_{iy}^{(2)}(r)}{H_{iy}^{(1)}(r)} \right).
\end{equation}
Then, we use the following elementary facts (see the Appendix, section A.1):
\begin{equation}
H_{iy}^{(2)}(r) = \overline{H_{-iy}^{(1)}(r)} = e^{-y\pi} \overline{H_{iy}^{(1)}(r)}.
\end{equation}
Thus,
\begin{equation}
M(r,s,iy) = i \frac{\pi}{4} s \left( H_{iy}^{(1)}(s) H_{iy}^{(2)}(s) - e^{-y\pi} (H_{iy}^{(1)}(s))^2 \ \frac{\overline{H_{iy}^{(1)}(r)}}{H_{iy}^{(1)}(r)} \right).
\end{equation}
Then, using the first estimates in Proposition \ref{imaxis}, we obtain for $y \in [-1,0]$:
\begin{equation}
\mid M(r,s,iy) \mid \leq C,
\end{equation}
and using the second ones, we have for $y\leq -1$,
\begin{equation}
\mid M(r,s,iy) \mid \leq 2 {\sqrt{\frac{s} {\mid y\mid}}}.
\end{equation}

\vspace{0.2cm}\noindent
As a conclusion, we have proved the following estimate for $\nu \geq 0$ or $\nu =iy$ with $y \leq 0$:
\begin{equation}
\mid M(r,s,\nu) \mid \leq C_{\delta}\ (1+\mid \nu \mid)^{-\frac{\delta}{2}} \ s^{\frac{\delta+1}{2}},
\end{equation}
where $C_{\delta}$  does not depend on $r,s$ and $\nu$. Now, we follow the same strategy as for the case $s \leq 1$, setting
\begin{equation}
f(\nu):= (1+ \nu )^{\frac{\delta}{2}} \ M(r,s,\nu).
\end{equation}
This application is of order $1$ with infinite type in the fourth quadrant, bounded on his boundary by $C s^{\frac{\delta+1}{2}}$. Using the Phragmen-Lindel\'{o}f's theorem again, we obtain the Lemma as in the first case.
\end{proof}


\vspace{0.2cm}\noindent
As an application, we have the following result:

\vspace{0.2cm}\noindent

\begin{prop}\label{fpmrealaxis}
\hfil\break
Assume that $r^{\frac{\delta+1}{2}} q(r)$ satisfies the hypothesis $(H_2)$ for some $\delta \in ]0,1[$.  Then, there exists $C>0$ such that, for all $r >0$ and $\nu$ in the second and fourth quadrant with $\mid \nu \mid \geq 1$ for the case $(+)$), (resp. for all
$\nu$ in the first and third  quadrant with $\mid \nu \mid \geq 1$ for the case $(-)$) one has:
\begin{eqnarray*}
\mid f^{\pm}(r, \nu) \mid \ &\leq& \ exp\ \left(  C \ \int_r^{+\infty} \min (s, s^{\frac{\delta+1}{2}})\ \mid q(s) \mid \ ds\right) \ \mid f_0^{\pm}(r, \nu) \mid \ . \\
\end{eqnarray*}
\end{prop}

\vspace{0.2cm}
\begin{proof}
For the case $(+)$ and $\nu$ in the fourth quadrant with $\mid \nu \mid \geq 1$, we solve (\ref{intg}) by iteration. We set
\begin{eqnarray*}
g_0(r)  &=& 1,\\
g_{k+1}(r) &=& \int_r^{+\infty} M(r,s,\nu) \ q(s) g_k(s) \ ds.
\end{eqnarray*}
Clearly, by Lemma \ref{estimationM}, we have the following estimate:
\begin{equation}
\mid g_k (r) \mid \leq \frac{1}{k!} \ \left( C \  \mid \nu \mid^{-\frac{\delta}{2}} \ \int_r^{+\infty} \min (s, s^{\frac{\delta+1}{2}})\ \mid q(s) \mid \ ds \right)^k.
\end{equation}
Hence, ${\ds{g(r)=\sum_{k=0}^{+\infty} g_k(r)}}$ satisfies
\begin{equation}
\mid g(r) \mid \leq exp\ \left( C \ \int_r^{+\infty}  \min (s, s^{\frac{\delta+1}{2}})\ \mid q(s) \mid \ ds\right),
\end{equation}
which implies the Proposition for the case $(+)$ and $\nu$ in the fourth quadrant. We deduce the other cases  from a parity argument and using (\ref{conjJost}).
\end{proof}

\begin{rem}
It follows from the proof of Proposition \ref{fpmrealaxis} that for $\nu$ in the second or fourth quadrant,
\begin{equation}
\fp \sim \fpo \ \ ,\ \  \nu \rightarrow \infty,
\end{equation}
whereas for $\nu$ in the first and the third quadrant,
\begin{equation}
\fm \sim \fmo \ \ , \ \ \nu \rightarrow \infty.
\end{equation}
\end{rem}

\vspace{0.5cm}\noindent
We deduce from Proposition \ref{fpmrealaxis} the following important result:

\vspace{0.2cm}
\begin{prop}
\hfil\break
The Jost solutions $\fpm$ are of order $1$ with infinite type with respect to  $\nu \in \C$.
\end{prop}

\begin{proof}
From Lemma \ref{orderonezero1} and Proposition \ref{fpmrealaxis}, we see that $\fp$, (resp. $\fm$), is of order one and infinite type in the second and the fourth quadrant, (resp. in the first and the third quadrant). In the next section (see Proposition \ref{zeroJost}), we shall prove that $\beta(\nu)$ does not vanish in the fourth quadrant, so using (\ref{SFS}) we can write for such $\nu$:
\begin{equation}
\fm = \frac{\reg - \al \fp}{\be}
\end{equation}
Moreover, it follows from (\cite{Lo}, Eq. (80)) that $\al$ and $\be$ are of order $1$ with infinite type for $\Re \nu \geq 0$. So, as previously, using (\cite{Lev1}, Theorem 12, p.22) and  (\ref{estreg}), we deduce that
$\fm$ is (at most) of order one with infinite type in the fourth quadrant, and also in the second quadrant by a parity argument.
\end{proof}


\vspace{0.5cm}
\noindent
Let us finish this section by the following result which will be useful to study the localization of the Regge poles in the Section 7.
\vspace{0.2cm}

\begin{prop}\label{estimationfpm}
\hfill\break
Assume that the potential $q$ satisfies $(H_1)$ and has a compact support. For any $\delta >0$ small enough, there exists $C>0$ such that, for all $\nu$ large enough and $r>0$,
\begin{equation*}\label{equivsoljost}
\mid \fpm - \fpmo \mid \ \leq \ \frac{C}{\mid \nu \mid +1}\ \mid \fpmo\mid   \ \ {\rm{for}} \ \mid Arg\ \nu \mid \leq \pd -\delta .
\end{equation*}
\end{prop}

\begin{proof}
We deduce from Proposition \ref{unifBessel} that, for $r$ in a compact set, one has the uniform asymptotics:
\begin{equation}\label{uniform}
 v(r) = - \frac{1}{\sqrt{\pi}}  \ \Gamma (\nu) \ \left(\frac{r}{2}\right)^{-\nu+\half } \ (1+o(1)) \ ,\ \ \nu \rightarrow \infty, \ \mid Arg \ \nu \mid \leq \pd - \delta .
\end{equation}
\noindent
It follows that for $\nu$ large enough in this domain, $v(r) \not=0$ and we can follow exactly the same strategy as in Proposition \ref{fpmrealaxis}.

\vspace{0.2cm}\noindent
We emphasize that for a potential $q$ with a compact support, the variables $s \geq r$ belong to a compact set, so using (\ref{kernelM}) and (\ref{uniform}) again, we see that the previous Green kernel $M(r,s,\nu)$ satisfies for
$\mid Arg \ \nu \mid \leq \pd - \delta$, the following estimate:
\begin{equation}
\mid M(r,s,\nu) \mid \ \leq \ C \ \  s^{-2\Re \nu + 1} \ \int_r^s \ t^{2\Re \nu -1} \ dt \ \leq \  \frac{C}{\Re \nu +1}.
\end{equation}
So, as in Proposition \ref{fpmrealaxis}, and setting again ${\ds{g(r) = \frac{\fp}{\fpo}}}$,  we obtain for $\mid Arg \ \nu \mid \leq \pd - \delta$,
\begin{equation}\label{controle}
\mid g(r) - 1  \mid \  \leq\ \frac{C}{\Re \nu +1}.
\end{equation}
Clearly, this implies the Proposition.
\end{proof}

\section{The Jost functions $\al$ and $\be$.}

In this Section, we recall first some well-known results for the Jost functions which can be found in \cite{Lo} for example. For convenience's reader,
we give the proofs since they are very simple and short.  In the second part of this Section, we shall establish some integral representations for the Jost functions.

\begin{lemma}\label{Jostimaginary}
\hfil\break
Assume that the potential $q$ satifies the hypotheses $(H_1)$ and $(H_2)$. For $y \in \R$, one has:
\begin{equation}
\mid \alpha(iy) \mid^2 - \mid \beta(iy) \mid^2 =y.
\end{equation}
\end{lemma}

\begin{proof}
Since $\nu \rightarrow \fpm$ are even functions, (\ref{SFS}) implies:
\begin{eqnarray}
\varphi(r,iy) &=& \alpha(iy) f^+(r,iy) + \beta(iy)f^-(r,iy), \\
\varphi(r,-iy) &=& \alpha(-iy) f^+(r,iy) + \beta(-iy)f^-(r,iy).
\end{eqnarray}
Moreover, it is easy to see that the Wronskian  $W(\varphi(r,iy), \varphi(r,-iy))= -2iy$, and  using (\ref{conjJostfunction}), one has:
\begin{equation}
\alpha(-iy) = \overline{\beta(iy)} , \ \beta(-iy) = \overline{\alpha(iy)}.
\end{equation}
Then, one obtains:
\begin{equation}
W \left( \alpha(iy) f^+(r,iy) + \beta(iy)f^-(r,iy), \  \overline{\beta(iy)} f^+(r,iy) + \overline{\alpha(iy)} f^-(r,iy) \right) = -2iy.
\end{equation}
The lemma follows immediately from (\ref{wronskien}).
\end{proof}

\begin{lemma} \label{alphabeta}
\hfil\break
Assume that the potential $q$ satisfies the hypotheses $(H_1)$ and $(H_2)$. For $\Re \mu \geq 0, \ \Re \nu \geq 0$ such that $\Re (\mu +\nu)>0$, one has:
\begin{equation}
2i \ \left( \al \beta(\mu) - \alpha(\mu) \be \right) = (\nu^2 -\mu^2) \ \int_0^{+\infty} \frac{\reg \varphi(r,\mu)}{r^2} dr.
\end{equation}
\end{lemma}

\begin{proof}
First, we remark that the integral  converges since $\Re (\mu +\nu)>0$. Secondly, using (\ref{Schradiale}), one has:
\begin{eqnarray}
\left( \varphi(r,\mu) \varphi '(r,\nu)- \varphi ' (r,\mu) \varphi (r,\nu) \right)'& = & \varphi (r,\mu) \varphi''(r,\nu)  - \varphi ''(r,\mu) \reg ,   \nonumber  \\
          &=& \varphi(r,\mu) \varphi (r,\nu) \ \frac{\nu^2-\mu^2}{r^2}. \label{deriv}
\end{eqnarray}

\noindent
Integrating (\ref{deriv}) onto $(0, +\infty)$, we obtain
\begin{eqnarray}
(\nu^2-\mu^2) \ \int_0^{+\infty} \frac{\varphi(r,\mu) \varphi (r,\nu)}{r^2} \ dr &=& \left[ W( \varphi(r,\mu)) , \reg \right]_{r=0}^{r=+\infty}, \nonumber \\
        &=& \left[ W(\varphi(r,\mu)), \reg \right]_{|r=+\infty}, \label{deriv1}
\end{eqnarray}
since $\Re (\mu +\nu)>0$. In order to calculate this last wronskian, we use (\ref{SFS}) again:
\begin{eqnarray*}
\varphi(r,\mu) &=& \alpha(\mu) f^+(r,\mu) + \beta(\mu) f^-(r,\mu), \\
\reg &=& \al \fp + \be \fm.
\end{eqnarray*}
Using (\ref{deriv1}) and the following elementary asymptotics, when $r \rightarrow + \infty$:
\begin{equation}
W(\fp, f^+(r,\mu)) \rightarrow 0 \ \ ,\ \ W(\fp, f^-(r,\mu)) \rightarrow -2i,
\end{equation}
the lemma is proved.
\end{proof}

\vspace{0.5cm}\noindent
Hence, we can deduce easily:

\begin{coro} \label{link}
\hfil\break
Assume that the potential $q$ satifies the hypotheses $(H_1)$ and $(H_2)$. For $\Re \nu >0$, one has:
\begin{equation}
\mid \al \mid^2 - \mid \be \mid^2 = 2 \ \Re \nu \ \Im \nu \ \int_0^{+\infty} \frac{\mid \reg \mid^2}{r^2} \ dr.
\end{equation}
\end{coro}

\begin{proof}
We take $\mu= \bar{\nu}$ in the previous lemma and one uses (\ref{conjreg}).
\end{proof}

\vspace{0.5cm}
Lemma \ref{alphabeta} and Corollary \ref{link} allow us to localize the zeros of the Jost functions. We recall that the first quadrant (resp. the fourth quadrant) of the complex plane is the set of the complex number $\nu$ such that $\Re \nu \geq 0$ and $\Im \nu \geq 0$,  (resp $\Re \nu \geq 0$ and $\Im \nu \leq 0)$. At least, the Regge poles are the complex zeros of the Jost function $\be$.

\vspace{0.2cm}
\begin{prop}\label{zeroJost}
\hfil\break
Assume that the potential $q$ satifies the hypotheses $(H_1)$ and $(H_2)$. Then, the Jost function $\al$, (resp. $\be$) does not vanish in the first quadrant, (resp. the fourth quadrant). In other words, the Regge poles belong to the first quadrant.
\end{prop}

\begin{proof}
Since $\al = \overline{\beta(\bar{\nu})}$, we only study the zeros of the Jost function $\al$. For $\Re \nu >0$ and $\Im \nu >0$, Corollary \ref{link} implies that
$\alpha(\nu) \not=0$. In the same way, if $\nu =iy$ with $y\not=0$, using Lemma \ref{Jostimaginary}, we see that $\alpha(iy) \not=0$. At least, if $\nu \geq 0$,
we have $\be = \overline{\al}$ and (\ref{SFS}) implies
\begin{equation}
\reg = \al \fp + \overline{\al} \fm.
\end{equation}
It follows that $\al \not=0$.
\end{proof}

\vspace{1cm}
In the next Propostion, we give  integral representations for the difference of two Jost functions which are a slight generalization of (\cite{AlRe}, p. 38). We adopt the following rule: if $q$ and $\tilde{q}$ are two potentials, we use the notation $Z$ and $\tilde{Z}$ for all the relevant  scattering quantities relative to these potentials.

\begin{prop}\label{diffJost}
\hfil\break
Let $q$ and $\tilde{q}$ two potentials satisfying $(H_1)$ and $(H_2)$.  For $\Re \nu \geq 0$, one has:
\begin{eqnarray}
\al -  \tal & = & \ \frac{1}{2i} \ \int_0^{+\infty} (q (r)-\tilde{q}(r)) \ \fm \ \treg \ dr. \\
\be  - \tbe  & = & - \frac{1}{2i} \ \int_0^{+\infty} (q (r)-\tilde{q}(r)) \ \fp \  \treg\ dr.
\end{eqnarray}
\end{prop}

\begin{proof}
We follow the same strategy as in Lemma \ref{alphabeta}. Using (\ref{Schradiale}), one has:
\begin{equation}
\left( \fm \tilde{\varphi}' (r,\nu) -  f^{-'}(r,\nu)\treg \right)' =  (\tilde{q}(r) -q(r))\ \fm \treg. \label{derivJost}
\end{equation}

\noindent
Integrating (\ref{derivJost}) onto $(0, +\infty)$, we obtain
\begin{equation}
\left[  W( \fm, \treg ) \right]_{r=0}^{r=+\infty} = \int_0^{+\infty} (\tilde{q}(r) -q(r))\ \fm \treg\ dr.
\end{equation}
When $r \rightarrow +\infty$, $\fm \sim \tfm$ (and also for the derivatives). It follows from (\ref{Jostfunctionalpha}) that
$W( \fm, \treg) \rightarrow 2i \tal$. In the same way, when $r\rightarrow 0$, $\treg \sim \reg$. Thus, as previously, one has $W( \fm, \treg) \rightarrow 2i \al$.
\end{proof}

\vspace{0.2cm}
As a consequence, we have the following integral representation which is the key point to prove our local uniqueness inverse result in Theorem  \ref{Mainresult}:

\vspace{0.2cm}
\begin{prop}\label{differenceJost}
\hfil\break
Let $q$ and $\tilde{q}$ be two potentials satisfying $(H_1)$ and $(H_2)$. Then, for $\Re \nu \geq 0$,
\begin{equation}
\al \tbe - \tal \be \ = \ \frac{1}{2i} \ \int_0^{+\infty} (q (r)-\tilde{q}(r)) \ \reg\ \treg \ dr.
\end{equation}
\end{prop}

\begin{proof}
We write:
\begin{equation}
 \al \tbe - \tal \be = (\al - \tal) \be - (\be - \tbe) \al.
\end{equation}
The result follows immediately from (\ref{SFS}) and Proposition \ref{diffJost}.
\end{proof}

\vspace{0.2cm}\noindent
We can also deduce from the previous Proposition the next technical result used in the proof of Theorem \ref{Mainresult}:

\begin{coro}\label{compact}
\hfil\break
Let $q$ and $\tilde{q}$ be two potentials satisfying $(H_1)$ and $(H_2)$. Assume also that $q = \tilde{q}$ a.e on $[a +\infty[$. Then, there exists $ C>0$ such that:
\begin{equation}
\mid \al \tbe - \tal \be \mid \ \leq \ \frac{C}{\Re \nu +1} \ a^{2\Re \nu} \ \ , \  \forall \ \Re \nu \geq 0.
\end{equation}
\end{coro}

\begin{proof}
For $r\leq a$, we know from (\ref{estreg}) that there exists $C>0$ such that, for all $\nu$ with $\Re \nu \geq 0$,
\begin{equation}
\mid \varphi (r,\nu) \mid \leq C \  r^{\Re \nu + \half} .
\end{equation}
and identically for $\treg$.  Then, applying Proposition \ref{differenceJost}, we obtain:
\begin{equation}
\mid \al \tbe - \tal \be \mid \ \leq \ C \  \int_0^a r^{2\Re \nu +1} \mid q(r)-\tilde{q}(r) \mid \ dr.
\end{equation}
Thus, using for instance Lemma 3.1 in  \cite{Ho1}), we obtain the Corollary.
\end{proof}

\vspace{0.5cm}\noindent
Now, roughly speaking, the following Proposition asserts that  the Jost functions $\al$ and $\be$ are
suitable perturbations of the free ones in the regime  $\nu \rightarrow +\infty$, ($\nu$ real), when the potential decays as $O(r^{-\rho})$ with $\rho>\frac{3}{2}$ at infinity.

\begin{prop}\label{equivJost}
\hfil\break
Let $q(r)$ be a potential satisfying $(H_1)$. Assume also that $r^{\frac{1+\delta}{2}} q(r)$ satisfies $(H_2)$ for some $\delta \in (0,1)$. Then,
\begin{equation}
\al \sim \alo \ \ ,\ \ \be \sim \beo \ \ \ {\rm{when}} \ \nu \rightarrow + \infty.
\end{equation}
\end{prop}

\begin{proof}
For instance, let us show $\be \sim \beo$ when $\nu \rightarrow +\infty$. We use Proposition \ref{diffJost} with the potential $\tilde{q}=0$:
\begin{equation}
\be - \beo =    \frac{1}{2i} \ \int_0^{+\infty} q(r) \ f^+(r,\nu) \  \varphi_0(r,\nu) \ dr.
\end{equation}
We recall that:
\begin{equation*}
\beo = \half \ A(\nu) \ e^{i(\nu +\half) \pd}.
\end{equation*}
Hence, one obtains:
\begin{equation}
\mid \frac {\be}{\beo} -1 \mid \ \leq \ \int_0^{+\infty} \mid q(r) \fp \frac{\varphi_0(r,\nu)}{A(\nu)} \mid \ dr.
\end{equation}
Using Proposition \ref{fpmrealaxis}, we obtain:
\begin{eqnarray*}
\mid \frac {\be}{\beo} -1 \mid &\leq& C \ \int_0^{+\infty} \mid q(r) \fpo \frac{\varphi_0(r,\nu)}{A(\nu)} \mid \ dr, \\
                                &\leq& C \ \int_0^{+\infty} \mid rq(r) H_{\nu}^{(1)}(r) J_{\nu}(r) \mid \ dr,
\end{eqnarray*}
where we have used (\ref{phi0}) and (\ref{fp0}).  Thus, it follows from Proposition \ref{unifBessel} and Theorem \ref{estproduitbessel} that:
\begin{eqnarray*}
\mid \frac {\be}{\beo} -1 \mid \ &\leq& C \ \left( \int_0^1 \mid rq(r) H_{\nu}^{(1)}(r) J_{\nu}(r) \mid \ dr +  \int_1^{+\infty} \mid rq(r) H_{\nu}^{(1)}(r) J_{\nu}(r) \mid \ dr \right) \\
&\leq& \frac{C}{\nu} \ \int_0^1 r \mid q(r) \mid \ dr \ + \ C_{\delta}\  \nu^{-\frac{\delta}{2}} \ \int_1^{+\infty} \mid q(r) \mid r^{\frac{\delta+1}{2}} \ dr,
\end{eqnarray*}
for some $\delta \in (0,1)$, which implies the Proposition.
\end{proof}

\section{The generalized phase shifts $\delta(\nu)$.}

In this section, we give some properties of the generalized phase shifts $\delta(\nu)$ for complex values of the angular momentum $\nu$. We recall that they are defined for $\nu>0$ by the formula:
\begin{equation}\label{defphase}
e^{2i \delta(\nu)} = e^{i\pi(\nu +\half)} \ \frac{\al}{\be},
\end{equation}
using the convention $\delta(\nu) \rightarrow 0$ when $\nu \rightarrow +\infty$.  Of course, in order to define properly $\delta(\nu)$ for complex variables $\nu$, we have to ensure that $\al$ and $\be$ do not vanish in a simply connected domain. We recall that the zeros of $\al$ belong to the fourth quadrant, whereas the zeros of $\be$, called the Regge poles, are located  in the first quadrant.

\vspace{0.2cm}\noindent
{\it{Definition: we say that a potential $q(r)$ satisfies the property $(R)$ if there exists $A \geq 0$ such that there are no Regge poles in the simply connected domain :}}
\begin{equation}
\Gamma_A= \{ \nu \in \C\ ;\ \Re \nu > A \}.
\end{equation}

\vspace{0.2cm} \noindent
Reminding that $\overline{\al} = \beta(\bar{\nu})$, we see that (\ref{defphase}) allows us to define $\delta(\nu)$ as an holomorphic function on this domain. We shall give in the next section some examples of such potentials.

\vspace{0.5cm} \noindent
The first property obtained in this Section has been observed by M. Horvath in \cite{Ho1}, but as we said previously, we think that his argument is not correct since, in general, the phase shifts $\delta(\nu)$ are not well defined in the presence of Regge poles. So, it is necessary to assume that the property $(R)$ is satisfied.

\vspace{0.2cm}\noindent
Now, let us recall some useful facts on holomorphic functions of the complex variable $z$.

\vspace{0.2cm} \noindent
{\it{Definition}}: A function $f(z)$ that is holomorphic in the half-upper plane $\Im z >0$ and takes its values in the half-upper plane is called a {\it{Herglotz function}}.

\vspace{0.2cm} \noindent
A Herglotz function has a nice growth property (see \cite{Lev1}, Theorem 8):
\begin{equation}\label{growth}
\forall z, \ \Im z > 0, \quad \mid f(z) \mid \ \leq \ 5 \mid f(i)\mid  \frac{\mid z\mid^2}{\Im z}.
\end{equation}

\vspace{0.2cm} \noindent
We deduce  immediately from Corollary \ref{link} the following result:

\begin{prop}\label{herglotz}
\hfill\break
Let $q(r)$ be a potential satisfying $(H_1)$, $(H_2)$ and the property $(R)$.
Then, the function $\delta(\nu)- \pd \nu$ is Herglotz in the variable $z= -(\nu-A)^2$, $\nu \in \Gamma_A,\ \Im \nu<0$.
\end{prop}

\begin{proof}
Let $\nu$ be a complex number in the fourth quadrant with $\nu \in \Gamma_A$. Corollary \ref{link} implies
\begin{equation}
\left| \frac{\al}{\be} \right| \ < 1.
\end{equation}
So, using (\ref{defphase}), we obtain:
\begin{equation}
\left| e^{2i (\delta(\nu) -(\nu+\half) \pd)} \right| \ < 1,
\end{equation}
or equivalently $\Im (\delta(\nu) - \pd \nu ) > 0$.
\end{proof}

\vspace{0.5cm} \noindent
The second property  was cited in \cite{Ho1} for potentials such that $rq(r)$ satisfies $(H_2)$.
We generalize this result to potentials $q(r)$ which has a slower decay at infinity. For simplicity,
we assume here that $q(r)$ is regular at $r=0$, but we can certainly allow some singularity at the origin. Of course, as previously, we need to assume that the property $(R)$ is satisfied.

\vspace{0.2cm}
\begin{prop}\label{deltanupoly}
\hfill\break
Let $q(r)$ be a potential satisfying the property $(R)$.
We also assume that $\mid q(r) \mid \leq C \ (1+r)^{-\rho}$ with $\rho> \frac{3}{2}$  for all $r>0$.
Then, for $\nu \in \Gamma_A$, there exists $C>0$ such that:
\begin{equation}
\mid \delta (\nu) \mid \ \leq \ C \ \mid \nu \mid^{4}.
\end{equation}
\end{prop}

\begin{proof}
We follows the same strategy as in \cite{Ho1}, Section 3.  We start from Proposition \ref{newrep}:
\begin{equation}\label{rappel0}
\reg = -2i\be \ e^{-i \pd (\nu-\half)} \ u(r) + \int_0^{+\infty} K(r,s, \nu) q(s)  \varphi(s,\nu) \ ds.
\end{equation}
We define the set:
\begin{equation}
\Omega = \{ \nu \in \Gamma_A \ :\ \ \mid \Im \nu \mid  <1 \ \}.
\end{equation}
Using Corollary \ref{estGreenK}, we see that for any $\delta\in (0,1)$ and $\nu \in \Omega$,
\begin{equation}
\mid K(r,s,\nu) \mid \ \leq \ C \ \mid \nu \mid^{-\frac{\delta}{2}} \ (rs)^{\frac{\delta+1}{4}},
\end{equation}
where the constant $C$ depends implicitly of $\delta$. We deduce from (\ref{rappel0}):
\begin{equation}\label{premiere}
\mid \reg + 2i\be \ e^{-i \pd (\nu-\half)} \ u(r) \mid \ \leq \ C \ \mid \nu \mid^{-\frac{\delta}{2}} \ r^{\frac{\delta+1}{4}} \
\int_0^{+\infty}  s^{\frac{\delta+1}{4}} \mid q(s)  \varphi(s,\nu)\mid \ ds,
\end{equation}
and also:
\begin{equation}\label{seconde}
\mid \reg \mid \ \leq \ 2 \mid \be e^{-i \pd (\nu-\half)} \ u(r) \mid + C \ \mid \nu \mid^{-\frac{\delta}{2}} \ r^{\frac{\delta+1}{4}} \
\int_0^{+\infty}  s^{\frac{\delta+1}{4}} \mid q(s)  \varphi(s,\nu)\mid \ ds.
\end{equation}
We multiply (\ref{seconde}) by $r^{\frac{\delta+1}{4}} \mid q(r)\mid$ and we integrate over $(0,+\infty)$:
\begin{eqnarray*}\label{troisieme}
\int_0^{+\infty}  r^{\frac{\delta+1}{4}} \mid q(r)  \varphi(r,\nu)\mid \ dr &\leq&  C \ \mid \be \mid \int_0^{+\infty}  r^{\frac{\delta+1}{4}} \mid q(r) u(r) \mid \ dr \\
& & \hspace{-0.5cm} + \ C   \ \mid \nu \mid^{-\frac{\delta}{2}} \ \int_0^{+\infty}  r^{\frac{\delta+1}{2}} \mid q(r) \mid \ dr  \cdot \int_0^{+\infty}  s^{\frac{\delta+1}{4}} \mid q(s)  \varphi(s,\nu)\mid \ ds.
\end{eqnarray*}
By our hypothesis, if we choose $\delta>0$ small enough, the integral $\int_0^{+\infty}  r^{\frac{\delta+1}{2}} \mid q(r) \mid \ dr$ is convergent, thus:
\begin{eqnarray}
\int_0^{+\infty}  r^{\frac{\delta+1}{4}} \mid q(r)  \varphi(r,\nu)\mid \ dr &\leq&  C \ \mid \be \mid \int_0^{+\infty}  r^{\frac{\delta+1}{4}} \mid q(r) u(r) \mid \ dr \nonumber  \\
& &  + \ C   \ \mid \nu \mid^{-\frac{\delta}{2}} \  \int_0^{+\infty}  s^{\frac{\delta+1}{4}} \mid q(s)  \varphi(s,\nu)\mid \ ds.
\end{eqnarray}
Hence, for $\nu \in \Omega$ large enough, one obtains:
\begin{equation}\label{quatre}
\int_0^{+\infty}  r^{\frac{\delta+1}{4}} \mid q(r)  \varphi(r,\nu)\mid \ dr \ \leq\   C \ \mid \be \mid \int_0^{+\infty}  r^{\frac{\delta+1}{4}} \mid q(r) u(r) \mid \ dr.
\end{equation}
Putting (\ref{quatre}) into (\ref{premiere}), and recalling that $\mid \Im \nu \mid$ is bounded, we have:
\begin{eqnarray}
\mid \reg + 2i\be \ e^{-i \pd (\nu-\half)} \ u(r) \mid \ &\leq & \ C \ \mid \nu \mid^{-\frac{\delta}{2}} \ r^{\frac{\delta+1}{4}} \
\mid -2i \be e^{-i \pd (\nu-\half)} \mid \nonumber \\
& & . \  \int_0^{+\infty}  s^{\frac{\delta+1}{4}} \mid q(s) u(s) \mid \ ds,
\end{eqnarray}
or equivalently,
\begin{equation}\label{cinq}
\left| \frac{\reg}{ -2i \be e^{-i \pd (\nu-\half)}} -u(r) \right| \leq C \mid \nu \mid^{-\frac{\delta}{2}} \ r^{\frac{\delta+1}{4}} \
\int_0^{+\infty}  s^{\frac{\delta+1}{4}} \mid q(s) u(s) \mid \ ds,
\end{equation}
since $\be$ does not vanish in $\Omega$. On the other hand, Proposition \ref{newrep} asserts:
\begin{equation}\label{rappel1}
\al \ e^{i \pd (\nu-\half)} + \be \ e^{-i \pd (\nu-\half)} \ =\ -\int_0^{+\infty} u(r) q(r) \varphi(r,\nu) \ dr.
\end{equation}
Dividing (\ref{rappel1}) by $\be e^{-i \pd (\nu-\half)}$ and using (\ref{defphase}), we obtain:
\begin{equation}
e^{2i\delta(\nu)} -1 = -2i \ \int_0^{+\infty} u(r) q(r) \frac{\reg}{ -2i \be e^{-i \pd (\nu-\half)}}\ dr.
\end{equation}
Thus,
\begin{eqnarray}
e^{2i\delta(\nu)} -1 &=& -2i \ \int_0^{+\infty} u^2(r) q(r) \ dr \nonumber \\
 & & -2i \ \int_0^{+\infty} u(r) q(r) \left( \frac{\reg}{ -2i \be e^{-i \pd (\nu-\half)}}-u(r) \right) \ dr.
\end{eqnarray}
Thus, we deduce from (\ref{cinq}) that:
\begin{equation}
\left| e^{2i\delta(\nu)} -1 + 2i \ \int_0^{+\infty} u^2(r) q(r) \ dr \right| \ \leq \ C \
\mid \nu \mid^{-\frac{\delta}{2}} \ \left( \int_0^{+\infty}  r^{\frac{\delta+1}{4}} \mid q(r) u(r) \mid \ dr \right)^2 .
\end{equation}
Using again that, for  $\delta>0$ small enough, the integral ${\ds{\int_0^{+\infty}  r^{\frac{\delta+1}{2}} \mid q(r) \mid \ dr}}$ is convergent, the Cauchy-Schwartz's inequality implies:
\begin{equation}
\left| e^{2i\delta(\nu)} -1 + 2i \ \int_0^{+\infty} u^2(r) q(r) \ dr \right| \ \leq \ C \ \mid \nu \mid^{-\frac{\delta}{2}}  \
\int_0^{+\infty}   \mid q(r) \mid \ \mid  u(r) \mid^2 \ dr .
\end{equation}
It follows that, for $\nu \in \Omega$:
\begin{eqnarray}\label{estimationps}
\left| e^{2i\delta(\nu)} -1  \right| \ &\leq& \ C \  \int_0^{+\infty}   \mid q(r) \mid \ \mid  u(r) \mid^2 \ dr, \nonumber \\
&\leq& \ C \   \ \int_0^{+\infty}   \mid r q(r) \mid \ \mid J_{\nu}(r) \mid^2 \ dr .
\end{eqnarray}
Now, by our hypothesis, we use the following estimate ${\ds{\mid r q(r) \mid \leq  \frac{C}{\sqrt{r}} }}$ for all $r>0$, and we obtain:
\begin{equation}
\left| e^{2i\delta(\nu)} -1  \right| \ \leq \ C \
\int_0^{+\infty}   \frac{ \mid  \ J_{\nu}(r) \mid^2}{\sqrt{r}} \ dr.
\end{equation}
This last integral can be estimated using Corollary \ref{majorationint}; for all $\nu \in \Omega$,
\begin{equation}
\int_0^{+\infty}   \frac{ \mid  \ J_{\nu}(r) \mid^2}{\sqrt{r}}  \ dr \ \leq \ C \ \mid \nu \mid^{-\half}.
\end{equation}
It follows that for $\nu \in \Omega$,
\begin{equation}
\left| e^{2i\delta(\nu)} -1  \right| \ \leq \ C \ \mid \nu \mid^{-\half}.
\end{equation}
We deduce that:
\begin{equation}
\delta(\nu) = k(\nu)\pi + \epsilon(\nu),
\end{equation}
with $k(\nu) \in \Z$ and $\epsilon(\nu) = O(\mid \nu \mid^{-\half})$. For $\nu$ large enough in $\Omega$, $\mid \epsilon(\nu) \mid  <\pi$, hence $\nu \rightarrow k(\nu)$ is a continuous function which implies that $k(\nu)$ is constant for $\nu \in \Omega$ large enough. Since  $\delta(\nu) \rightarrow 0$ as $\nu \rightarrow +\infty$, this constant is equal to zero, and  we have obtained:
\begin{equation}\label{dansomega}
\delta(\nu) = O(\mid \nu \mid^{-\half}) \ \ {\rm{for}} \ \nu \in \Omega.
\end{equation}
Now, let us consider  complex numbers $\nu \in \Gamma_A$ with $\Im \nu \leq -1$. Using (\ref{growth}) and Proposition \ref{herglotz}, one has:
\begin{equation}\label{loindeomega}
\mid \delta(\nu) -\pd \nu \mid \ \leq C \  \frac{\mid \nu-A \mid^{4}}{\mid \Im (\nu-A)^{2} \mid} \ \leq C \ \mid \nu-A \mid^{4}.
\end{equation}
Thus, using that ${\overline{\delta(\nu)}} = \delta(\bar{\nu})$ (which follows from the definition of the generalized phase shifts), and using (\ref{dansomega}), (\ref{loindeomega}), we have:
\begin{equation}
\mid \delta(\nu) \mid \ \leq C \   \mid \nu \mid^{4} \ \ {\rm{for}} \ \nu \in \Gamma_A.
\end{equation}
\end{proof}

\vspace{0.5cm}\noindent
As a by-product of the proof of Proposition \ref{deltanupoly}, we can give some estimates on the generalized phase shifts $\delta(\nu)$, when $\nu \rightarrow + \infty$ and for the class of potentials with super-exponential decay at $+\infty$. Propositions \ref{deltanupoly}  and \ref{super} will be very useful later to prove
the existence of an infinite number of Regge poles.

\vspace{0.5cm}\noindent
{\it{Definition:}}
A function $f :\ ]0, +\infty[ \rightarrow \R$ is super-exponentially decreasing if for any $A \geq 0$, there exists $C>0$ (depending on $A$) such that:
\begin{equation}
\mid f(r) \mid \leq C \  e^{-Ar} \ \ {\rm{for\ all}} \ r>0.
\end{equation}

\vspace{0.3cm}\noindent
We have the following result which is very close to \cite{Ho1}, Corollary 1, Eq. (14). We shall use this Proposition in the proof of Theorem \ref{nombreinfini}.

\vspace{0.2cm}
\begin{prop}\label{super}
\hfill\break
Let $q(r)$ be a potential such that ${\ds{\mid q(r)\mid  \ \leq\  C\ e^{-Ar},\ \forall r >0}}$. Then, for all $B<A$,
\begin{equation*}
\delta(\nu) = O \left( \frac{1}{\sqrt{\nu}} \ e^{-\nu \ Argcosh \ (1+\frac{B^2}{2})} \right) \ ,\ \nu \rightarrow +\infty.
\end{equation*}
In particular, if the potential $q(r)$ is  super-exponentially decreasing, the generalized phase shifts $\delta(\nu)$ are super-exponentially decreasing, ($\nu$ real).
\end{prop}

\begin{proof}
By hypothesis, for any $B<A$, there exists $C>0$ such that
\begin{equation}\label{pot}
\mid q(r) \mid \ \leq \ \frac{C}{r} \ e^{-Br} \ , \ \forall r > 0.
\end{equation}
So, using (\ref{estimationps}), we have for $\nu>0$ large enough:
\begin{eqnarray*}\label{estimationsps1}
\left| e^{2i\delta(\nu)} -1  \right| \  &\leq& \ C \ \left( \ \int_0^{+\infty}   e^{-Br} \ \mid J_{\nu}(r) \mid^2 \ dr  \right), \\
                                        &\leq  & C \ Q_{\nu -\half}^0 (1+ \frac{B^2}{2}),
\end{eqnarray*}
where $Q_{\nu -\half}^0 $ is the Legendre function of the second kind (see \cite{Ol}, (10.22.6)). But, when $\nu \rightarrow +\infty$, we know (see \cite{Ol}, (14.3.10), (14.15.14) and (10.25.3)) that:
\begin{eqnarray} \label{legendre}
Q_{\nu -\half}^0 (\cosh \eta) & = &    \sqrt{\frac{\eta}{\sinh( \eta)}} \ K_0 (\nu \eta) \ (1+o(1)), \nonumber \\
       &=&    \sqrt{\frac{\pi}{2 \nu \sinh(\eta)}} \ e^{-\nu \eta} \ (1+o(1)).
\end{eqnarray}
Then, the  proposition taking $\cosh \eta = 1+ \frac{B^2}{2}$.
\end{proof}

\section{Localization of the Regge poles.}

\subsection{The case of super-exponentially decreasing potentials.}

In this section, using Propositions \ref{deltanupoly} and \ref{super}, we prove that for a non-zero super-exponentially decreasing potential, the number of Regge poles is always infinite, and moreover their real parts tend to infinity in the first quadrant. As we have said in the Introduction, this theoretical result contradicts Theorem 5.2 in \cite{HBM}, which says that for an integrable potential $q(r)$ on $(0,+\infty)$, there are finitely many Regge poles in the right-half plane. We emphasize that our theoretical result will be confirmed in the next subsection, where for potentials with compact support, the asymptotics of the Regge poles are given.

\begin{thm}\label{nombreinfini}
\hfill\break
Let $q(r)$ be a non-zero potential satisfying $(H_1)$ and which is super-exponentially decreasing. \par\noindent
Then, the  number of Regge poles is infinite and their real parts tend to infinity in the first quadrant.
\end{thm}

\begin{proof}
Assume that the Regge poles are contained in a vertical strip in the first quadrant , i.e there exists $A>0$ such that $\be \not=0$ for $\Re \nu > A$, (i.e for
$\nu \in \Gamma_A$). Then $\delta(\nu)$ can be defined as an holomorphic function on $\Gamma_A$ and Proposition
\ref{deltanupoly} asserts that for $\Re \nu >A$, there exists $C>0$ such that:
\begin{equation}\label{majdelta}
\mid \delta (\nu) \mid \ \leq C \ \mid \nu \mid^4.
\end{equation}

\vspace{0.2cm}\noindent
Now, let us recall an elementary result for functions of the complex variable, belonging to the Hardy class.
The Hardy class $H_+^2$,  (see for instance \cite{Lev}, Lecture 19) is the set of analytic functions $F$ in the right half-plane
$\Omega = \{ \nu \in \C \ ,\ \Re \nu >0 \}$, satisfying the condition
\begin{equation}\label{Hardy0}
\sup_{x >0} \ \int_{\R} \ \mid F(x+i y) \mid ^2 \ dy \ < \infty,
\end{equation}
and equipped with the norm
\begin{equation}\label{norme}
\mid \mid F \mid \mid = \left( \sup_{x >0} \ \int_{\R} \ \mid F(x+i y) \mid ^2 \ dy  \right)^{\frac{1}{2}}.
\end{equation}

\vspace{0.2cm}\noindent
The Paley-Wiener Theorem asserts that a function $F(\nu)$ belongs to the Hardy space $H_+^2$ if and only if there exists a function $f \in L^2(0,+\infty)$ such that
\begin{equation}\label{PaleyWiener}
F(\nu) = \frac{1}{\sqrt{2\pi}}\ \int_0^{+\infty} e^{-t\nu} \ f(t) \ dt \ ,\ \forall \nu \in \Omega.
\end{equation}
Moreover, we have:
\begin{equation} \label{normeL2}
\mid \mid F \mid \mid \  = \   \mid \mid f \mid \mid_{L^2(0,\infty)}.
\end{equation}


\vspace{0.2cm}\noindent
A function $F(\nu)$ belonging in the Hardy class has a very nice property: if one knows that $F(\nu)= O\left( e^{-B\nu} \right), \ \nu \rightarrow +\infty$,  one can show that $F(\nu)$ satisfies a {\it{uniform bound}} for $\nu \in \Omega$. Actually, we emphasize that we have a better result; it suffices to have the previous estimate for $\nu$ integer, (see \cite{DGN}, Proposition 4.2). We prefer to give here this result in this discrete setting since we shall use it again in the next Section.

\vspace{0.2cm}\noindent

\begin{prop}\label{Hardy}
\hfill \break
Let $F$ be a function in the Hardy class $H_+^2 $. Assume that for some $B>0$, we have $F(l) = O\left( e^{-Bl} \right), \ l \rightarrow +\infty$, ($l$ integer).
Then,
\begin{equation} \label{majoration}
\mid F(\nu) \mid  \ \leq \  \frac{\mid \mid F \mid \mid}{\sqrt{4\pi \Re \nu}}  \ e^{-B  \Re  \nu} \ , \ \forall \nu \in \Omega.
\end{equation}
\end{prop}

\vspace{0.2cm}\noindent
We use Proposition \ref{Hardy}  with the function $F(\nu)$ defined by:
\begin{equation}
F(\nu) = \frac{\delta(\nu +A+1)}{(\nu +A+1)^5},
\end{equation}
which belongs to the Hardy class thanks to (\ref{majdelta}).
Now, we use the fact that $q(r)$ is super-exponentially decreasing. Proposition \ref{super} implies : $\forall B>0$,
\begin{equation}
F(\nu) = O \left( e^{-B \nu} \right) \ \ ,\ \ \nu \rightarrow + \infty.
\end{equation}
Hence, we deduce from Proposition \ref{Hardy} that $F(\nu) =0$ since $B>0$ is arbitrary. It follows from (\ref{interpolation}) and (\ref{phaseshifts}) that:
\begin{equation} \label{equa}
\al = e^{-i(\nu +\half) \pi} \ \be  \ \ {\rm{for \ all\ }} \Re \ \nu \geq A+1,
\end{equation}
and by analytical continuation, (\ref{equa}) holds for $\Re \nu>0$. Then the Regge interpolation $\sigma (\nu)=1$ for all $\Re \nu >0$, so
applying the Loeffel's uniqueness Theorem, (see \cite{Lo}, Theorem 2), we obtain $q = 0$. Note we could also use Novikov's results to obtain $q=0$, (see \cite{No1}, \cite{No2}), since  the potential $q(r)$ is exponentially decreasing.

\end{proof}





\subsection{Potentials with compact support.}

 Now, let us study the case of a potential having compact support. We have seen above that the Regge poles are always in infinite number and their real parts go to infinity. For potentials with compact support, we can improve our previous result. The first proposition (which is certainly known) obtained in this Section shows that the Regge poles concentrate (in some sense given below) on the positive imaginary axis. This theoretical result will be confirmed in the second Proposition where the precise asymptotics for the Regge poles and for a large class of potentials with compact support are obtained.

\vspace{0,2cm}\noindent

 \begin{prop}\label{nombrefini}
 \hfill\break
 Let $q$ be a potential satisfying $(H_1)$ and with compact support. Then, for all $\delta>0$, there are a finite number of Regge poles in the sector
 $Arg \ \nu \in [0, \pd-\delta]$.
 \end{prop}

\begin{proof}
Let $q$ be a potential with support in $[0,b]$. By Proposition \ref{diffJost} with $\tilde{q}=0$, we have:
\begin{equation}
\be - \beo = \frac{1}{2i} \ \int_0^b q(r) \fp \rego \ dr.
\end{equation}
So,
\begin{equation}
\frac{\be}{\beo}-1 = \frac{1}{2i} \ \int_0^b q(r) \fp \frac{\rego}{\beo} \ dr.
\end{equation}
It follows from (\ref{Anu}), (\ref{phi0}) and (\ref{betazero}) that:
\begin{equation}
\frac{\rego}{\beo} = \sqrt{2\pi r} \ e^{-i(\nu+\half)\pd} \ J_{\nu}(r),
\end{equation}
then,
\begin{equation}
\left| \frac{\rego}{\beo} \right| \ \leq  \  \sqrt{2\pi r}\  e^{\Im \nu \pd} \  \mid J_{\nu}(r) \mid.
\end{equation}
Now, using Proposition \ref{estimationfpm}, we see that there exists $C>0$ such that:
\begin{equation}
\mid \fp \mid \leq C \mid \fpo \mid \ \ {\rm{for}} \  Arg \ \nu \in [0, \pd-\delta].
\end{equation}
It follows from (\ref{fp0}) that:
\begin{equation}
\mid \fpm \mid \leq C \sqrt{r} \ e^{-\Im \nu \pd}\ \mid H_{\nu}^{(1)}(r) \mid.
\end{equation}
Then,
\begin{equation}
\left| \frac{\be}{\beo}-1 \right| \leq C \ \int_0^b r \mid q(r) \mid \ \mid J_{\nu}(r) H_{\nu}^{(1)}(r)\mid \ dr .
\end{equation}
Since $r$ belongs to a compact set, Proposition \ref{unifBessel} implies:
\begin{equation}
\left| \frac{\be}{\beo}-1 \right| \leq \frac{C}{\mid \nu \mid +1},
\end{equation}
and thus, for $\nu$ large with $Arg \ \nu \in [0, \pd-\delta]$, $\be \not=0$.
\end{proof}

\vspace{0,5cm}\noindent
Now, we are able to give the precise asymptotics of the Regge poles for a potential with compact support. As we shall see, the Regge poles concentrate (in a certain sense)  along the positive imaginary axis. We begin by a first technical Lemma. In this Lemma, in order to simplify the proof, $q$ is assumed to be $C^2$ on his support.

\vspace{0,2cm}\noindent

\begin{lemma}
\hfill \break
Let $q(r)$ be a piecewise continuous potential having his support in $[0,a]$. We assume that $q$ is $C^2$ in $[0,a]$ and let $\delta>0$ small enough.
Then, for $Arg \ \nu \in [\pd -\delta, \pd]$, one has as $\nu \rightarrow \infty$:
\begin{equation} \label{lemmetechnique}
\frac{\be}{\beo} = \left[ 1 -\frac{2i\pi}{(\nu+1) \Gamma^2(\nu +1)} \ \left( \frac{a}{2} \right)^{2\nu+2} \
\left( q(a-0) + O\left(\frac{1}{\nu}\right) \right) \right] \ + \ O\left(\frac{1}{\nu}\right).
\end{equation}
\end{lemma}

\begin{proof}
We use the following integral representation for the regular solution proved in \cite{Ra}: there exists an integral kernel $R(r,s)$ {\it{independent of $\nu$} } such that
\begin{equation}\label{RammInt}
\reg = \rego + \int_0^r R(r,s) \ \varphi_0(s,\nu) \ \frac{ds}{s^2},
\end{equation}
where
\begin{enumerate}
\item $R(r,s)$ is $C^2$ with respect to $r,s$ : $0<r<+\infty, 0<s\leq r$.
\item ${\ds{R(r,r) = \frac{r}{2} \int_0^r s q(s) \ ds}}$.
\item $R(r,0)=0.$
\item There exists $C>0$ such that for all $r > 0$, ${\ds{ \int_0^r \mid R(r,s) \mid \ \frac{ds}{s} \ \leq \ C}}$.
\item $K(r,s)$ has first derivatives which are bounded.
\end{enumerate}
Note in particular that a simple application to the mean value Theorem gives us the estimate
\begin{equation} \label{Rs0}
  R(r,s) = O(s), \quad s \to 0.
\end{equation}
\vspace{0.2cm}\noindent

Now, Proposition \ref{diffJost} with $q=0$ and $\tilde{q} = q$ gives:
\begin{equation}
\be -\beo = \frac{1}{2i} \ \int_0^a \  q(r) \fpo \ \reg \ dr.
\end{equation}
So, using (\ref{fp0}) and (\ref{betazero}), one has:
\begin{equation}
\frac{\be}{\beo} = 1 -i \int_0^a \ \sqrt{\frac{\pi r}{2}} q(r) \ H_{\nu}^{(1)} (r) \ \frac{\reg}{A(\nu)} \ dr.
\end{equation}
It follows from (\ref{phi0}) and (\ref{RammInt}) that:
\begin{eqnarray}\label{fraction}
\frac{\be}{\beo} &=& 1 - \frac{i\pi}{2} \int_0^a \ r q(r) \ J_{\nu}(r) \ H_{\nu}^{(1)} (r)  \ dr \nonumber \\
&& - \frac{i\pi}{2} \int_0^a \ {\sqrt{r}} q(r)\ \int_0^r s^{-\frac{3}{2}}\  R(r,s) \ J_{\nu}(s) \ H_{\nu}^{(1)} (r) \ ds \ dr.
\end{eqnarray}
For $Arg \ \nu \in [\pd -\delta, \pd]$ and for $r,s$ in a compact set, one has the following uniform asymptotics (see Corollary \ref{uniformeprecis})
when $\nu \rightarrow \infty$:
\begin{eqnarray}\label{termedominant}
J_{\nu}(s) \ H_{\nu}^{(1)} (r) &= & \frac{1}{i\pi \nu} \left( \frac{s}{r}\right)^{\nu} \ \left(1+\frac{r^2 -s^2}{4\nu} +O\left(\frac{1}{\nu^2}\right)\right)  \nonumber \\
& &\  + \ \frac{2}{\Gamma^2 (\nu+1)} \left(\frac{rs}{4}\right)^{\nu} \left(1- \frac{r^2 +s^2}{4\nu} +O\left(\frac{1}{\nu^2}\right) \right) .
\end{eqnarray}
Now, using (\ref{termedominant}), we can estimate the first term in the (RHS) of (\ref{fraction}). One has:
\begin{eqnarray*}
\int_0^a  r q(r) \ J_{\nu}(r)  H_{\nu}^{(1)} (r)  \ dr &=& \int_0^a \ r q(r) \left[ \frac{1}{i\pi \nu} \left( 1+ O\left(\frac{1}{\nu^2}\right) \right)\right.\\
& & \left. + \ \frac{2}{\Gamma^2 (\nu+1)} \left(\frac{r}{2}\right)^{2\nu} \left( 1-\frac{r^2}{2\nu}  +O\left(\frac{1}{\nu^2}\right) \right) \right] \ dr \\
&=& \frac{2^{1-2\nu}}{\Gamma^2 (\nu+1)} \int_0^a \ r^{2\nu+1} q(r)  \left( 1 -\frac{r^2}{2\nu} +  O\left(\frac{1}{\nu^2}\right) \right) \ dr + O\left(\frac{1}{\nu}\right)
\end{eqnarray*}
Let us examine the above integral. For instance, integrating twice by parts, we see easily that:
\begin{equation}
\int_0^a \ r^{2\nu+1} q(r)\ dr = \frac{a^{2\nu +2}}{2\nu+2} \left( q(a-0)+ O\left(\frac{1}{\nu}\right) \right),
\end{equation}
and so on for the other terms. Thus, we obtain easily:
\begin{equation}\label{ligne1}
\int_0^a  r q(r) \ J_{\nu}(r)  H_{\nu}^{(1)} (r)  \ dr = \frac{4}{(\nu+1)\Gamma^2 (\nu+1)} \left( \frac{a}{2} \right)^{2\nu+2} \
\left( q(a-0)  + O\left(\frac{1}{\nu}\right) \right)   +\ O\left(\frac{1}{\nu}\right) .
\end{equation}
Similarly, we claim that:
\begin{eqnarray} \label{ligne2}
\int_0^a  {\sqrt{r}} q(r)\ && \hspace{-0.6cm} \int_0^r s^{-\frac{3}{2}}  R(r,s)  J_{\nu}(s)  H_{\nu}^{(1)} (r) \ ds \ dr =
O\left( \frac{1}{\nu^2\Gamma^2 (\nu+1)} \left( \frac{a}{2} \right)^{2\nu+2} \right) +  O\left(\frac{1}{\nu^2}\right)
\end{eqnarray}
For instance, using (\ref{termedominant}), we need to estimate:
$$
	I_1 = \frac{1}{i\pi \nu} \int_0^a r^{\frac{1}{2} - \nu} q(r) \int_0^r s^{-\frac{3}{2} + \nu} R(r,s) ds dr.
$$
Using (\ref{Rs0}) and an integration par parts, we see that
\begin{equation} \label{Rs1}
  \int_0^r s^{-\frac{3}{2} + \nu} R(r,s) ds = \frac{r^{\nu - \frac{1}{2}}}{\nu - \frac{1}{2}} \left( R(r,r) + O \left(\frac{1}{\nu} \right) \right).
\end{equation}
Hence we get easily
\begin{equation}
  I_1 =  \frac{1}{i\pi \nu (\nu - \frac{1}{2})} \int_0^a q(r) R(r,r) dr \ + \, O \left(\frac{1}{\nu^2} \right)  = O \left(\frac{1}{\nu^2} \right).
\end{equation}
We can estimate the other terms similarly and we leave the details to the reader.
\end{proof}

\vspace{1cm}\noindent
We use the previous Lemma to give the  main result of this Section. When the potential $q(r)$ has a discontinuity on the boundary of his support, i.e when $q(a-0)\not=0$, we can calculate precisely the asymptotics of the Regge poles in the first quadrant. We have in mind the example of the square well potential defined by $q(r)=q_0$ if $r \in ]0,a]$ and $q(r)=0$ for $r>a$.
Let us denote by $\nu_p$ the Regge poles in the first quadrant and assume that they are ordered according to their increasing modulus. We shall see that the leading term of this asymptotic expansion does not depend on the depth and the width of the potential, and that the Regge poles concentrate on the positive imaginary axis in the following meaning:
\begin{equation}
\frac{\Im \nu_p}{\Re \nu_p} = \frac{2\log p}{\pi}\ (1+o(1)) \ \ {\rm{ when}}\ \  p \rightarrow +\infty.
\end{equation}

\vspace{0.2cm}\noindent
\begin{thm}
\hfill \break
Let $q(r)$ be a piecewise continuous potential having his support in $[0,a]$. We assume that $q$ is $C^2$ in $[0,a]$ and $q(a-0) \not=0$. Then, the Regge poles $\nu_p$, $p \in \N$ satisfy:
\begin{equation}
\nu_p =  \left( \frac{p\pi^2}{2\log^2 p} +i \ \frac{p\pi}{\log p} \right) (1+o(1)) \ \ ,\ \ p \rightarrow +\infty.
\end{equation}
\end{thm}

\begin{proof}
We give only the main ingredients of the proof and we leave the details to the reader. For $Arg \ \nu \in [\pd-\delta, \pd]$, we define:
\begin{equation}
f(\nu) \ =\ 1 -\frac{2i\pi}{(\nu+1) \Gamma^2(\nu +1)} \ \left( \frac{a}{2} \right)^{2\nu+2} \ q(a-0).
\end{equation}
Using Stirling's formula, we obtain easily:
\begin{equation}\label{asymptw}
f(\nu) = 1 -iq(a-0)\ e^{-2w[\log w -1 -\log(\frac{a}{2})]} \ \left( 1+ O\left(\frac{1}{w}\right) \right),
\end{equation}
where we have set $w=\nu+1$. Now, if we set $A =1+ \log(\frac{a}{2})$ and $z = e^{-A}w$, we can write the equation (\ref{asymptw}) as:
\begin{equation}
f(\nu)= 1 -iq(a-0) \ e^{-2e^Az \log z} \ \left( 1+ O\left(\frac{1}{z}\right) \right).
\end{equation}
For instance, we assume that $q(a-0)>0$ and we define the function $g(z)$ by:
\begin{equation}
g(z) = 1 -iq(a-0) \ e^{-2e^Az \log z}.
\end{equation}
We obtain immediately that the zeros $z_p$ of the function $g(z)$ must verify, for $p \in \Z$, the equation:
\begin{equation} \label{lambert0}
z_p \log z_p = \frac{1}{2e^A} \ \left( \log q(a-0) +i (2p+\half) \pi \right) := \alpha_p .
\end{equation}
Since we are looking for the Regge poles in the first quadrant, we only have to consider the case where $p \in \N$. Setting $z_p = e^{u_p}$ in (\ref{lambert0}) we see that we have to study the roots $u_p$ of the following equation:
\begin{equation}\label{lambert}
u_p\ e^{u_p} = \alpha_p,
\end{equation}
which is the so-called Lambert's equation. The solution is given by $u_p =W(\alpha_p)$ where $W$ stands for the Lambert's function.
Using the asymptotics of the Lambert function (for the principal branch), (see for instance \cite{De}, or \cite{CG}, Eq. $(4.20)$), we obtain:
\begin{equation}
u_p = \log \alpha_p - \log_2 \alpha_p +o(1) \ \ ,\ \ p \rightarrow + \infty.
\end{equation}
So, we deduce:
\begin{equation}
z_p = \frac {\alpha_p}{\log \alpha_p} \ (1+ \ o(1)).
\end{equation}
An easy calculation gives:
\begin{equation}
\log \alpha_p = \log \left( \frac{p}{2e^A}\right) +i \pd +o(1),
\end{equation}
thus it follows that:
\begin{eqnarray}
\Re z_p \ &=&  \ \frac{p \pi^2}{2 e^A(\log p)^2}\ (1+o(1)), \\
\Im z_p \  &=& \ \frac{p \pi}{e^A \log p}\ (1+o(1)).
\end{eqnarray}
Now, the end of the proof follows from a standard application of Rouché Theorem. We refer to (\cite{FrYu}, pp. 35-36) for the details.
\end{proof}

\subsection{The case of potentials dilatable analytically.}

In this section, we study the localization of the Regge poles for short-range potentials $q(r)$ which can be extended analytically in the complex angular sector of the positive real axis $\mid \arg \ z \mid \leq b\leq \pd$.
In particular, the results of this section generalize to the case $b =\pd$ those results of Barut and Diley \cite{BaDi} who proved that the Regge poles are confined into a domain which is contained in a vertical strip on the first quadrant, (see also \cite{Bes}). In order to make complex scaling, we introduce the following formalism:

\vspace{0.2cm}\noindent
\begin{defi}
Let $\mathcal{V}$ a small complex neighborhood of $0 \in \C$. We say that a potential $q(r)$ is dilation-analytic short range if $q(r)$ can be extended analytically in a conic neighborhood of $(0,+\infty)$ with:
\begin{equation}\label{dilation}
\mid q(e^{\theta} r) \mid \ \leq C \ (1+r)^{-\rho} \ ,\ \rho>1 \ , \ {\rm{for \ all}} \ \theta \in \mathcal{V} .
\end{equation}
\end{defi}

\vspace{0.5cm}\noindent
In other words, if the disc $D(0,b) \subset \mathcal{V}$ with $b \leq \pd$, the above definition means that $q(r)$ can be extended analytically in the complex angular sector of the positive real axis  $\mid \arg \ z \mid \leq b\leq \pd$. For simplicity, we shall say that the potential $q(r)$ belongs to the class $\mathcal{A}_b$.

\vspace{0.2cm}\noindent
Now, let us consider radial Schrödinger equations with the energy $\lambda=k^2$:
\begin{equation}\label{Schrodingerenergie}
\left( - \frac{d^2}{dr^2} + \frac{\nu^2-\frac{1}{4}}{r^2} +q(r) \right) \ u(r) = k^2 u(r).
\end{equation}

\vspace{0.2cm}\noindent
We emphasize that in this section, the dependance with respect to the energy and the potential is important since we shall make later complex scalings. The regular solution is denoted by $\varphi(r,\nu; k, q)$ and satisfy:
\begin{equation}
\varphi(r,\nu; k, q) \ \sim \  r^{\nu + \half} \ ,\ r \rightarrow 0.
\end{equation}
If we solve the integral equation (\ref{repint1}) taking into account the energy $k^2$, we obtain easily that, for $\Re \nu \geq 0$,  the regular solution $\varphi(r,\nu, k, q)$ is analytic with respect to $k^2$, (for details, see \cite{AlRe}, p. 19).

\vspace{0.2cm}\noindent
Similarly, the Jost solutions are denoted by $f^{\pm}(r,\nu;k,q)$ and satisfy:
\begin{equation} \label{asympcomplexe}
f^{\pm}(r,\nu;k,q)\  \sim \ e^{\pm i k r} \ ,\ r \rightarrow + \infty.
\end{equation}
As for the regular solution, taking into account the energy $k^2$ in the integral equation (\ref{integralequationjost}), we can show that, for $\Re \nu \geq 0$, the Jost solution
$f^+(r,\nu;k,q)$ is analytic with respect to $k$ for $\Im k \geq 0$, whereas $f^-(r,\nu;k,q)$ is analytic for $\Im k \leq 0$, (see \cite{AlRe}, p. 24).

\vspace{0.2cm}\noindent
Now, as in Section 2, we define the Jost functions:
\begin{eqnarray}\label{jostcomplexe}
\alpha(\nu;k,q) &=& \ \frac{i}{2} \ W( \varphi(r,\nu; k, q), f^-(r,\nu;k,q) ), \nonumber \\
\beta(\nu;k,q) &=& - \frac{i}{2} \ W( \varphi(r,\nu; k, q), f^+(r,\nu;k,q) ).
\end{eqnarray}
Clearly, $\alpha(\nu;k,q)$ is defined for $\Re \nu \geq 0$ and $\Im k \leq 0$, whereas $\beta(\nu;k,q)$ is defined for $\Re \nu \geq 0$ and $\Im k \geq 0$.

\vspace{0.5cm}\noindent
In order to exploit the fact that our potential is dilatable analytically, we begin with the following elementary lemma whose easy proof is omitted:

\begin{lemma}\label{scaling}
\hfill \break
For $\theta \in \R$, the regular solution and the Jost solutions satisfy:
\begin{eqnarray}
\varphi(e^{\theta} r,\nu; k, q) & =& e^{(\nu+\half)\theta} \ \varphi ( r,\nu; e^{\theta} k, q_{\theta}), \\
f^{\pm}(e^{\theta}  r,\nu;k,q)  & =& f^{\pm}(  r,\nu;e^{\theta} k,q_{\theta}),
\end{eqnarray}
where $q_{\theta}(r) = e^{2\theta}\ q(e^{\theta}r)$.
\end{lemma}

\vspace{0.3cm}\noindent
Hence, for dilation-analytic potentials in  $\mathcal{V}$, Lemma \ref{scaling} allows us to define analytically $\varphi(e^{\theta} r,\nu; k, q)$ for $\theta \in \mathcal{V}$, and $f^{\pm}(e^{\theta}  r,\nu;k,q)$ for $\theta \in \mathcal{V}$ with the condition $\pm  \Im (e^{\theta}k) \geq 0$.

\vspace{0.2cm}\noindent
We can deduce the following result for the Jost function $\be= \beta(\nu;1,q)$:

\begin{coro}\label{scaling1}
\hfill \break
Assume that the potential $q(r)$  is dilation-analytic short range. Then, for $\theta \in \mathcal{V}$, one has:
\begin{equation}
\beta(\nu;1, q) = e^{(\nu-\half) \theta} \ \beta(\nu; e^{\theta}, q_{\theta}).
\end{equation}
\end{coro}

\begin{proof}
Using Lemma \ref{scaling}, one has for $\theta \in \R$,
\begin{eqnarray*}
\varphi ( r,\nu; e^{\theta} , q_{\theta}) &=& e^{- (\nu+\half)\theta} \ \varphi(e^{\theta} r,\nu; 1, q), \\
f^+(  r,\nu;e^{\theta} ,q_{\theta}) &=& f^+(e^{\theta}  r,\nu;1,q).
\end{eqnarray*}
Using (\ref{jostcomplexe}), one obtains
\begin{eqnarray}
\beta(\nu; e^{\theta}, q_{\theta}) &=& -\frac{i}{2} \ W ( e^{- (\nu+\half)\theta} \ \varphi(e^{\theta} r,\nu; 1, q), f^+(e^{\theta}  r,\nu;1,q)), \\
                           &=& e^{- (\nu+\half)\theta} \ e^{\theta} \ \beta(\nu; 1,q).
\end{eqnarray}
Then, the result follows from a standard analytic continuation.

\end{proof}

\vspace{0.5cm}\noindent
Now, we can establish our main result concerning the localization of the Regge poles:


\begin{thm}\label{absencepoles}
\hfill \break
Let $q(r)$ be a potential belonging to the class $\mathcal{A}_b$ for some $b \in]0, \pd]$ with $\rho>1$.
Then, there exists $A>0$ such that there are no Regge poles in $\Gamma^b_{A}= \{ \nu \in \C\ ;\ \Re \nu >A,\ \mid \arg(\nu-A)\mid \leq b \}$.
\end{thm}

\vspace{0.5cm}\noindent
Or course, this result is really pertinent  when $0\leq \arg(\nu-A) \leq b$ since we know that there are no Regge poles in the fourth quadrant.
Moreover, if we take $b =\pd$ in Theorem \ref{absencepoles}, we obtain immediately the following Corollary which has been proved in \cite{BaDi} by Barut and Diley  in the case $\rho>2$ only.

\begin{coro}\label{pasdepoles}
\hfill\break
Let $q(r)$ be a potential wich can be extended analytically in $\Re z \geq 0$ and such that
\begin{equation}
\mid q(z)\mid \ \leq \ C\ (1+\mid z \mid)^{-\rho} \ \ ,\ \ \rho>1.
\end{equation}
Then, the Regge poles are bounded to the right in the first quadrant.
\end{coro}

\begin{proof}
We follow in spirit and simplify the approach given in (\cite{AlRe}, pp. 80-82) for  Yukawian potentials, (i.e for exponentially decreasing potentials).
We define:
\begin{equation}
\mathcal{O}_b = \{ \nu \in \C, \ Arg \ \nu \in [0,b] \},
\end{equation}
and we denote by $\mathcal{R}_b$ the set of the Regge poles belonging to $\mathcal{O}_b$. To simplify the notation, we set:
\begin{eqnarray}
\Psi(r) &=& \varphi(r,\nu, e^{ib}, q_{ib}), \\
W(r) &=& q_{ib}(r).
\end{eqnarray}
Let $\nu \in \mathcal{R}_b$ be a Regge pole for the potential $q(r)$ at the energy $\lambda=1$, i.e we assume that $\be=\beta(\nu;1,q)=0$. First, we remark that $\Re \nu >0$ since $\be$ does not vanish on the imaginary axis. We shall use implicitly this fact to give sense of the next integrals.  Secondly, using Corollary \ref{scaling1}, we get $\beta(\nu; e^{ib}, W)=0$. Thus, using (\ref{SFS}) and a standard analytic continuation argument, we have:
\begin{eqnarray}
\Psi(r) &=& \alpha(\nu; e^{ib}, W)\  f^+(r, \nu; e^{ib}, W)   + \beta(\nu; e^{ib}, W) \ f^-(r, \nu; e^{ib}, W), \nonumber \\
        &=& \alpha(\nu; e^{ib}, W)\  f^+(r, \nu; e^{ib}, W).
\end{eqnarray}
We emphasize that it follows from (\ref{asympcomplexe}) that $\Psi(r)$, (and its derivative), decay exponentially when $r \rightarrow +\infty$. Now, we start from:
\begin{equation}\label{eqpsi}
-\Psi'' + \left( \frac{\nu^2-\frac{1}{4}}{r^2} +W(r) \right) \ \Psi = e^{2ib} \ \Psi.
\end{equation}
We multiply (\ref{eqpsi}) by $\overline{\Psi}$ and we integrate by parts on $(0,+\infty)$. We obtain:
\begin{equation}
\int_0^{+\infty} \mid \Psi' \mid^2 \ dr - \left[ \Psi' \overline{\Psi} \right]_0^{+\infty} + \int_0^{+\infty} \left( \frac{\nu^2-\frac{1}{4}}{r^2} +W(r) \right) \
\mid \Psi \mid^2 \ dr = \int_0^{+\infty} e^{2ib} \ \mid \Psi \mid^2 \ dr.
\end{equation}
When $r \rightarrow + \infty$, $\Psi' \overline{\Psi}$ is exponentially decreasing, and when $r \rightarrow 0$, $\Psi' \overline{\Psi} =O(r^{2\Re \nu})$, with $\Re \nu >0$. Hence, we get:
\begin{equation}\label{prem}
\int_0^{+\infty} \mid \Psi' \mid^2 \ dr  + \int_0^{+\infty} \left( \frac{\nu^2-\frac{1}{4}}{r^2} +W(r) \right) \
\mid \Psi \mid^2 \ dr = \int_0^{+\infty} e^{2ib} \ \mid \Psi \mid^2 \ dr.
\end{equation}
\vspace{0.1cm}\noindent
Now, let us remark that:
\begin{equation}\label{carre}
\int_0^{+\infty} \left| \Psi' - \frac{\Psi}{2r} \right|^2 \ dr = \int_0^{+\infty} \mid \Psi ' \mid^2 - \Re (\frac{\Psi' \overline{\Psi}}{r})
+ \frac{\mid \Psi \mid^2}{4r^2} \ dr .
\end{equation}
Integrating by parts on $(0,+\infty)$, one has easily:
\begin{equation}
\int_0^{+\infty} \frac{\Psi' \overline{\Psi}}{r} \ dr = \int_0^{+\infty} \frac{\mid \Psi \mid^2}{r^2} \ dr - \int_0^{+\infty} \frac{\Psi \overline{\Psi'}}{r} \ dr,
\end{equation}
so we obtain:
\begin{equation}\label{partiereelle}
2 \Re \left( \int_0^{+\infty} \frac{\Psi' \overline{\Psi}}{r} \ dr \right) = \int_0^{+\infty} \frac{\mid \Psi \mid^2}{r^2} \ dr.
\end{equation}
Putting (\ref{partiereelle}) into (\ref{carre}), we have:
\begin{equation} \label{carre1}
\int_0^{+\infty} \left| \Psi' - \frac{\Psi}{2r} \right|^2 \ dr = \int_0^{+\infty} \mid \Psi ' \mid^2 - \frac{\mid \Psi \mid^2}{4r^2} \ dr .
\end{equation}
Thus, using (\ref{prem}) and (\ref{carre1}), one obtains:
\begin{equation}\label{carre2}
\int_0^{+\infty} \left[ \frac{\nu^2}{r} - re^{2ib} \right] \frac{\mid \Psi \mid^2}{r} \ dr = - \int_0^{+\infty} \left| \Psi' - \frac{\Psi}{2r} \right|^2 \ dr
- \int_0^{+\infty} W \mid \Psi \mid^2 \ dr
\end{equation}
Now,  multiplying (\ref{carre2}) by $ e^{i (\pd-b -Arg \ \nu)}$ and taking the real part, one has:
\begin{equation}
\mid \nu \mid \ \cos (\pd-b +Arg \ \nu)\ \int_0^{+\infty} \left( \frac{\mid \nu \mid}{r}
+ \frac{r}{\mid \nu \mid} \right) \frac{\mid \Psi \mid^2}{r} \ dr = - \int_0^{+\infty} \Re \left( e^{i(\pd-b -Arg \ \nu)} W \right) \mid \Psi \mid^2 \ dr \nonumber
\end{equation}
\begin{equation}\label{carre3}
-\cos (\pd -b -Arg\ \nu)\ \int_0^{+\infty} \left| \Psi' - \frac{\Psi}{2r} \right|^2 \ dr.
\end{equation}
Since $b \in ]0,\pd]$, $\cos (\pd -b -Arg\ \nu) \geq 0$, thus, we obtain immediately:
\begin{eqnarray}
\mid \nu \mid \ \cos (Arg\ \nu +\pd -b)\ \int_0^{+\infty} \left( \frac{\mid \nu \mid}{r}
+ \frac{r}{\mid \nu \mid} \right) \frac{\mid \Psi \mid^2}{r} \ dr && \nonumber \\
&&\hspace{-2cm} \leq - \int_0^{+\infty} \Re \left( e^{i(\pd-b -Arg \ \nu)} W \right) \mid \Psi \mid^2 \ dr \nonumber \\
&&\hspace{-2cm} \leq  \int_0^{+\infty} \mid W \mid \mid \Psi \mid^2 \ dr .
\end{eqnarray}
Our main hypothesis on the potential $q(r)$ implies a fortiori ${\ds{\mid W(r) \mid \leq \frac{C}{r}}}$, thus
\begin{equation}
\mid \nu \mid \ \cos (Arg\ \nu +\pd -b)\ \int_0^{+\infty} \left( \frac{\mid \nu \mid}{r}
+ \frac{r}{\mid \nu \mid} \right) \frac{\mid \Psi \mid^2}{r} \ dr \leq C \ \int_0^{+\infty}  \frac{\mid \Psi \mid^2}{r} \ dr .
\end{equation}
Reminding we are looking for the Regge poles in $\mathcal{O}_b$, one sees that $\cos (Arg\ \nu +\pd -b) \geq 0$, so using that
$\frac{\mid \nu \mid}{r} + \frac{r}{\mid \nu \mid} \geq 2$, one has:
\begin{equation}
\mid \nu \mid \ \cos (Arg\ \nu +\pd -b)\ \int_0^{+\infty} \frac{\mid \Psi \mid^2}{r} \ dr \leq \frac{C}{2}  \int_0^{+\infty} \frac{\mid \Psi \mid^2}{r} \ dr ,
\end{equation}
so we have
\begin{equation}
\mid \nu \mid \ \cos (Arg\ \nu +\pd -b)\  \leq \frac{C}{2},
\end{equation}
or equivalently
\begin{equation}
\Re \left( \nu \ e^{i(\pd-b)} \right)  \leq \frac{C}{2}.
\end{equation}
The Theorem follows easily.

\end{proof}

\section{Proof of the Theorem \ref{analyticalcase}.}

The goal of this section is to prove our Theorem \ref{analyticalcase}. Our proof is self-contained, elementary, and very close in spirit with the celebrated local  Borg-Marchenko's uniqueness Theorem, (see  \cite{Be, GS, Si, Te}). In particular, we emphasize that we do not use the Regge-Loeffel's uniqueness theorem, ( \cite{Lo}, Theorem 2).

\subsection{Uniqueness of the Regge interpolation function.}

Let us consider potentials $q(r)$, (resp. $\tilde{q}(r)$),  belonging to the class $\mathcal{A}$, i.e each potential can be extended analytically
in the domain  $\Re z \geq 0$ and such that, for all $z$ in this domain, we have (for instance for the potential $q(r)$):
\begin{equation*}
\mid q(z) \mid \ \leq \ C \  (1+ \mid z \mid)^{-\rho}\ \ , \ \ \rho > \frac{3}{2}.
\end{equation*}

\vspace{0.2cm}\noindent

Using Corollary \ref{pasdepoles}, there exists $N\in \N$ large enough such that $\be$ and $\tbe$ do not vanish for $\Re \nu >N-1$. So, we can define in this region the generalized phase shifts $\delta(\nu)$ and $\tilde{\delta}(\nu)$. Moreover, using Proposition \ref{deltanupoly}, for all $\nu$ with $\Re \nu \geq N$,
\begin{equation}\label{polynomial}
\mid \delta(\nu) \mid \leq C \mid \nu \mid^4 \ ,\ \mid \ \tilde{\delta}(\nu) \mid \leq C \mid \nu \mid^4.
\end{equation}

\noindent
It follows from (\ref{polynomial}) that the function $F(\nu)$ given by
\begin{equation}\label{defF}
F(\nu) = \frac{\delta(\nu + \frac{n-2}{2}+ N) - \tilde{\delta}(\nu  + \frac{n-2}{2}+ N)} {(\nu+N)^5}
\end{equation}
belongs to the Hardy class $H_+^2$, (see section 7 for the definition), and from our main hypothesis on the phase shifts, we have for all $A>0$ and $l \in \N$,
\begin{equation}
F(l) = O(e^{-Al}) \ ,\ l \rightarrow + \infty.
\end{equation}

\vspace{0.2cm}\noindent
Now, using Proposition \ref{Hardy} and since $A>0$ is here arbitrary, we obtain that $F(\nu) = 0$ for all $\nu$ with $\Re \nu \geq 0$. Hence, for $\Re \nu >0$ large enough, we have
$\delta(\nu)= \tilde{\delta}(\nu)$. It follows from (\ref{phaseshifts}) that
\begin{equation}
\sigma(\nu)  = \tilde{\sigma} (\nu) \ \ {\rm{for}} \ \Re \nu >0 \  {\rm{large \ enough}},
\end{equation}
or equivalently, using (\ref{interpolation}) and (\ref{phaseshifts}), $\al \tbe - \tal \tbe =0$  for $\Re \nu >0$ large enough. By a standard analytic continuation, this last equality holds true for $\Re \nu >0$. Thus using again (\ref{interpolation}) and (\ref{phaseshifts}), we have obtained:
\begin{equation}\label{equalphase}
 \sigma(\nu)  = \tilde{\sigma} (\nu) ,
 \end{equation}
 for $\Re \nu >0$ (here  both functions are meromorphic).

\subsection{A new proof of the Regge-Loeffel's theorem.}

At this stage, we could use Regge-Loeffel's uniquenes Theorem, (see \cite{Lo}, Theorem 2) as a black box to obtain $q = \tilde{q}$. Nevertheless, we prefer to  give here another proof  which has the advantage to be very simple, short and self-contained. We emphasize we shall also use this new approach for the proof of Theorem \ref{Mainresult}. Moreover, as we will see at the end of this section, this strategy allows us to obtain a new Regge-Loeffel's theorem which is local in nature.

\vspace{0.2cm}\noindent
We follow an idea close to the local Borg Marchenko uniqueness Theorem, (see  \cite{Be, GS, Si, Te}). We fix $r>0$, and we define $F(r,\nu)$ as a function of the complex variable $\nu$ by:
\begin{equation}\label{definitionF}
F(r,\nu) = \fp \tfm - \fm\tfp.
\end{equation}

\vspace{0.2cm}\noindent
As we have seen in the Section 4, $F(r,\nu)$ is holomorphic on $\C$ with respect to $\nu$, is even, and of order $1$ with infinite type.
Moreover, Proposition \ref{fpmimaginary} implies that $F(r,\nu)$ is bounded  on the imaginary axis.

\vspace{0.2cm}\noindent
Now, we aim at showing that $F(r,\nu) \rightarrow 0$ when $\nu \rightarrow + \infty$. For $\nu \geq0$, as $\be \not=0$,  we can set:
\begin{equation}
\Phi (r,\nu) =\frac{\reg}{\be}.
\end{equation}
Clearly, using (\ref{SFS}), we get:
\begin{equation}
\fm = \Phi(r,\nu) - \frac{\al}{\be} \ \fp,
\end{equation}
and thus,
\begin{eqnarray}\label{mainrelation1}
F(r,\nu) & = & \tilde{\Phi}(r,\nu) \fp - \Phi(r,\nu) \tfp  \nonumber   \\
          &  & + \ \left( \frac{\al}{\be} -    \frac{\tal}{\tbe} \right) \ \fp \tfp.
\end{eqnarray}
So, using (\ref{interpolation}), we deduce:
\begin{eqnarray}\label{mainrelation2}
F(r,\nu) & = & \tilde{\Phi}(r,\nu) \fp - \Phi(r,\nu) \tfp  \nonumber   \\
          &  & + \ e^{-i\pi (\nu+\half)} \ \left( \sigma(\nu) - \tilde{\sigma}(\nu) \right) \ \fp \tfp.
\end{eqnarray}
Hence, by (\ref{equalphase}), we see that for $\nu \geq 0$, $F(r,\nu)$ can be written  as:
\begin{equation}
F(r,\nu) = \tilde{\Phi}(r,\nu) \fp - \Phi(r,\nu) \tfp   .
\end{equation}

\vspace{0.2cm}\noindent
For instance, let us examine  $\Phi(r,\nu) \tfp $.  Propositions \ref{fpmrealaxis} and \ref{equivJost} imply for $\nu \geq 0$:
\begin{equation}
\mid \Phi(r,\nu) \tfp \mid \ \leq \ C \ \mid \frac{\reg}{\beo }\mid  \ \mid \fpo \mid .
\end{equation}
Using (\ref{estreg}), we get for a fixed $r>0$ and for all $\nu \geq 0$:
\begin{equation}
\mid \reg \mid \leq \ C \ r^{\nu},
\end{equation}
and Proposition \ref{unifBessel} gives:
\begin{equation}
\mid \fpo \mid \ \leq \ C \ (\frac{r}{2})^{-\nu + \half} \ \Gamma(\nu).
\end{equation}
Thus, using (\ref{betazero}),  we obtain easily:
\begin{equation}
\Phi(r,\nu) \tfp  =O(\nu^{-1})\ \ {\rm{when}} \ \ \nu \rightarrow + \infty.
\end{equation}

\vspace{0.2cm}\noindent
At this stage, we have then proved that $F(r,\nu) \rightarrow 0$ when $\nu \rightarrow + \infty$. In particular, by parity, $F(r,\nu)$ is bounded on the real axis.
Applying the Phragmén-Lindelöf theorem (see \cite{Boa}, Theorem 1.4.2 for instance) in each quadrant of the complex plane, we see that $F(r,\nu)$ is bounded on $\C$, and so is constant by Liouville's Theorem. As the
limit is $0$ when $\nu \rightarrow +\infty$, we have  $F(r,\nu) =0$ for all $\nu \in \C$.  So, using (\ref{definitionF}), we have :
\begin{equation}\label{equalJost}
\fp \tfm = \fm\tfp \ ,\ \forall \nu \in \C, \ \forall r>0.
\end{equation}
For $\nu \in \R$ fixed, we remark that, for all $r >0$, $\fpm \not=0$.\footnote{We could also use the same strategy as in Lemma \ref{nonzero} to obtain this result.} Indeed, assume for instance that $\fp =0$ for some $r>0$.
Since ${\ds{\overline{\fm} = \fp}}$, we have also $\fm=0$ which contradicts $W(\fp, \fm)=-2i$.
Then, we can write (\ref{equalJost}) as
\begin{equation}
\frac{\fp}{\fm} = \frac{\tfp}{\tfm} \ ,\ \forall \nu>0, \ \forall r>0.
\end{equation}
Differentiating  and using that $W(\fp, \fm)=-2i$, it follows that $(\fm)^2 = (\tfm)^2$. We take the logarithmic derivative of this and we differentiate once more. We obtain:
\begin{equation}
\frac{(\fm)''}{\fm} = \frac{(\tfm)''}{\tfm},
\end{equation}
Using (\ref{Schradiale}), we deduce $q = \tilde{q}$, for all $r >0$.

\vspace{0.5cm}
\noindent 
As we have said in the beginning of this section, it is not difficult to see that the previous approach allows us to obtain a {\it{local}} Regge-Loeffel's theorem:

\begin{thm} \hfill\break
Let $q(r)$ and
$\tilde{q}(r)$ be two potentials satisfying $(H_1)$ and such that $r^{\rho}q(r)$, $r^{\rho} \tilde{q}(r)$ satisfy $(H_2)$ for some
$\rho > \half$. If $\sigma(\nu)-\tilde{\sigma}(\nu) = o \left( \nu \left(\frac{ae}{2\nu}\right)^{2\nu}  \right)$ when $\nu \rightarrow +\infty$, then
$q(r) = \tilde{q}(r)$ for almost all $r \in (a,+\infty)$.
\end{thm}

\begin{proof}
Under these hypotheses, using  (\ref{mainrelation2}) and Proposition \ref{fpmrealaxis}, we see that, for $r \geq a$, the function $F(r,\nu) = o(1)$ when $\nu \rightarrow +\infty$, and we conclude as previously.
\end{proof}

\vspace{0.2cm}
\noindent 
We also note that we shall prove in Proposition \ref{necessarycondition} the following result :
if $q(r) = \tilde{q}(r)$ for almost all $r \in (a,+\infty)$, then 
\begin{equation}
\sigma(\nu)-\tilde{\sigma}(\nu) = O \left( \frac{1}{\nu^2} \left(\frac{ae}{2\nu}\right)^{2\nu}  \right).
\end{equation}

\section{Proof of the Theorem \ref{Mainresult}.}

\subsection{Proof of $(A_2) \Longrightarrow (A_1)$.}

We begin with the following Proposition which proves the implication $(A_2) \Longrightarrow (A_1)$ of Theorem \ref{Mainresult}. We emphasize that here, we only use the fact that the potentials decay sufficiently rapidly  at infinity; in particular, we do not use explicitly that the potentials belong to the class $\mathcal{C}$. Note that this Proposition has been proved in \cite{Ho} using a variational approach for the generalized phase shifts. For convenience's reader, we present here a shorter proof.

\begin{prop}\label{necessarycondition}
\hfil\break
Let $q$ and $\tilde{q}$ be two potentials satisfying $(H_1)$ and such that $r^{\rho} q(r)$ and $r^{\rho}  \tilde{q}(r)$ satisfy $(H_2)$ for some $\rho > \half$. Assume that $q(r)=\tilde{q}(r)$ for almost all $r \in (a,+\infty)$. Then, the corresponding phase shifts satisfy:
\begin{equation}
\delta_l - \tilde{\delta}_l \ = \ O \left(\frac{1}{l^{n}} \  \left( \frac{ae}{2l}\right)^{2l}\right) \ \ ,\ \ l\rightarrow + \infty.
\end{equation}
\end{prop}

\begin{proof}
We use  the Regge interpolation  (\ref{interpolation}) for $\nu>0$:
\begin{equation}
e^{2i\delta(\nu)} - e^{2i\tilde{\delta}(\nu)} = e^{i\pi(\nu+\half)} \ \frac{\al \tilde{\be} - \tilde{\al}\be}{\be \tilde{\be}}.
\end{equation}
Now, Corollary \ref{compact}, Proposition \ref{equivJost} and  (\ref{betazero}) imply:
\begin{equation} \label{inequality}
\mid  e^{2i\delta(\nu)} - e^{2i\tilde{\delta}(\nu)} \mid  \ \leq \ C \ \frac{a^{2\nu}}{(\nu+1)\mid \beo\mid^2}
                                                            \leq  C \ \left( \frac{a}{2} \right)^{2\nu} \ \frac{1}{(\nu+1)\Gamma^2(\nu+1)}.
\end{equation}
By definition, $\delta(\nu)$, (resp. $\tilde{\delta}(\nu)$)  $\rightarrow 0$ when $\nu \rightarrow +\infty$. Hence, it follows from (\ref{inequality}) that
\begin{equation}
\mid \delta (\nu) -\tilde{\delta}(\nu) \mid \ \leq \  C \ \left( \frac{a}{2} \right)^{2\nu} \ \frac{1}{(\nu+1)\Gamma^2(\nu+1)}.
\end{equation}
Then, applying Stirling's formula with $\nu=\nu(l) = l + \frac{n-2}{2}$, we obtain the result.

\end{proof}

\subsection{Proof of $(A_1) \Longrightarrow (A_2)$.}

\subsubsection{Reduction to the analytic case.}

Let us consider two potentials $q$ and $\tilde{q}$ belonging to the class $\mathcal{C}$, i.e
$q = q_1 + q_2$ such that $q_1$ has compact support in $[0,b]$, and $q_2$ can be extended holomorphically in $\Re z \geq 0$. In the same way,
$\tilde{q} = \tilde{q}_1 + \tilde{q}_2$ with the same properties. For $j=1,2$, we denote $\delta_l^j$, (resp. $\tilde{\delta}_l^j$), the phases shifts corresponding to the potential $q_j$, (resp. $\tilde{q}_j$).

\vspace{0.2cm}\noindent
First, we prove the following elementary result. This Lemma permits us to reduce our proof to the analytic case of Theorem \ref{analyticalcase}.

\begin{lemma}\label{reduc}
\hfill\break
Let $q$ and $\tilde{q}$  be two potentials belonging to the class $\mathcal{C}$, assume that
\begin{equation}
\delta_l - \tilde{\delta}_l \ = \ o \left( \frac{1}{l^{n-3}} \ \left( \frac{ae}{2l}\right)^{2l}\right) \ \ ,\ \ l\rightarrow + \infty.
\end{equation}
We set $c = max\ (a,b)$. Then,
\begin{equation}
\delta_l^2 - \tilde{\delta}_l^2 \ = \ o \left( \frac{1}{l^{n-3}} \ \left( \frac{ce}{2l}\right)^{2l}\right) \ \ ,\ \ l\rightarrow + \infty.
\end{equation}
\end{lemma}

\begin{proof}
We write:
\begin{equation}
\delta_l^2 - \tilde{\delta}_l^2 = (\delta_l^2-\delta_l) + (\delta_l - \tilde{\delta}_l) + (\tilde{\delta}_l-\tilde{\delta}_l^2),
\end{equation}
and using Proposition \ref{necessarycondition}, one has:
\begin{equation}
(\delta_l^2-\delta_l) + (\tilde{\delta}_l-\tilde{\delta}_l^2) = o \left(\frac{1}{l^{n-3}} \ \left(\frac{be}{2l}\right)^{2l}\right),
\end{equation}
which implies the lemma.
\end{proof}

\subsubsection{End of the proof of Theorem \ref{Mainresult}.}

Let us consider two potentials $q$ and $\tilde{q}$ belonging to the class $\mathcal{C}$, i.e
$q = q_1 + q_2$ such that $q_1$ has compact support in $[0,b]$, and $q_2$ can be extended holomorphically in $\Re z \geq 0$. In the same way,
$\tilde{q} = \tilde{q}_1 + \tilde{q}_2$ with the same properties. We assume that:
\begin{equation*}
\delta_l - \tilde{\delta}_l \ = \ o \left( \frac{1}{l^{n-3}} \ \left( \frac{ae}{2l}\right)^{2l}\right) \ \ , \ \ l \rightarrow + \infty.
\end{equation*}
First, we apply Lemma \ref{reduc} and Theorem \ref{analyticalcase} and we get $q_2 (r) =\tilde{q}_2(r)$ for all $r >0$.

\vspace{0.2cm}\noindent
Now, as in the proof of Theorem \ref{analyticalcase}, we define for a fixed $r>0$,
\begin{equation}\label{definitionFb}
F (r,\nu) = \fp \tfm - \fm \tfp.
\end{equation}
As previously, the application $F (r,\nu)$ is holomorphic on $\C$ with respect to $\nu$, is even and is bounded on the imaginary axis. Moreover,  this application is of order one with infinite type.

\vspace{0.2cm} \noindent
Now, our goal is to  show that for $r \geq a$, $F(r,\nu) \rightarrow 0$ when $\nu \rightarrow + \infty$.  Hence, as in the second proof of Theorem \ref{analyticalcase}, we shall get $q(r) = \tilde{q}(r)$ almost everywhere for $r \geq a$, and Theorem \ref{Mainresult} will be proved.

\vspace{0.2cm}\noindent
As in the previous section, we define for $\nu \geq 0$,
\begin{equation}
\Phi (r,\nu) =\frac{\varphi (r,\nu)}{\beta(\nu)},
\end{equation}
and as in (\ref{mainrelation2}), we obtain
\begin{eqnarray}\label{mainrelation3}
F(r,\nu) & = & \tilde{\Phi} (r,\nu) \fp - \Phi_(r,\nu) \tfp  \nonumber   \\
          &  &  + \ \left( \frac{\al}{\be} -    \frac{\tal}{\tbe} \right) \ \fp \tfp.
\end{eqnarray}

\vspace{0.2cm}\noindent
At this stage, it is important to make a crucial remark; since $q$ and $\tilde{q}$ may have  compact supports, Theorem \ref{nombreinfini} asserts that the Regge poles associated to these potentials may have their real parts that tend to $+\infty$ in the first quadrant. It follows that we have to modify the strategy of the second proof of Theorem \ref{analyticalcase}.

\vspace{0.2cm}\noindent
Of course, as in the last Section, we have ${\ds{\Phi (r,\nu) \tfp - \tilde{\Phi}_(r,\nu) \fp}}= O(\nu^{-1})$ when $\nu \rightarrow + \infty$. Then,
\begin{eqnarray}\label{mainrelationb1}
F(r,\nu) &=& \left( \frac{\al}{\be} -    \frac{\tal}{\tbe} \right) \ \fp \tfp + O(\nu^{-1}), \nonumber \\
           &=& \left( \al \tbe -\tal \be  \right) \ \frac{\fp \tfp}{\be \tbe} + O(\nu^{-1}).
\end{eqnarray}

\noindent
First, we see that Propositions \ref{fpmrealaxis} and \ref{equivJost} imply for $\nu \geq 0$:
\begin{equation}
\left| \frac{\fp \tfp}{\be \tbe} \right| \ \leq \ C \ \left| \frac{\fpo}{\beo} \right|^2.
\end{equation}
Hence, using Proposition \ref{unifBessel}, we get:
\begin{equation}
\left| \frac{\fp \tfp}{\be \tbe} \right| \ \leq \  \frac{C}{(\nu+1)^2\ r^{2\nu}}.
\end{equation}

\vspace{0.2cm}\noindent
This suggests to set:
\begin{equation}
G(\nu) = \left( \tal  \be  - \al \tbe \right) \ \ \frac{1}{(\nu+1)^2\ r^{2\nu}},
\end{equation}
and we aim at proving that $G(\nu) \rightarrow 0$ when $\nu \rightarrow +\infty$. To show this result, we use  Cartwright's Theorem (\cite{Boa}, Theorem 10.2.1) which we recall here:

\begin{thm}
\hfill\break
Let $f(\nu)$ be holomorphic in $\Re \nu \geq 0$. Assume there exists $A,B>0$ such that:
\begin{equation*}
\mid f(\nu) \mid \ \leq \ C\ exp \left( A\  \Re \nu + B \mid \Im \nu \mid \right).
\end{equation*}
If $B<\pi$ and if $f(l) \rightarrow 0$ as $l \rightarrow + \infty$, ($l$ integer), then $f(\nu) \rightarrow 0$ as $\nu \rightarrow + \infty$.
\end{thm}

\vspace{0.2cm}\noindent
We apply this Theorem with the function $f(\nu) := G(\nu + \frac{n-2}{2})$. Since $q_2 = \tilde{q}_2$, $q-\tilde{q}$ is supporting in $[0,b]$, then  using Corollary \ref{compact}, one has:
\begin{equation}
\mid f(\nu) \mid \ \leq \ C \ \left( \frac{b}{r}\right)^{ \ 2 \Re \nu}.
\end{equation}
Now, let us estimate $f(l)=G(\nu(l))$. One starts from:
\begin{equation}
f(l) = \left( \frac{\tilde{\alpha}(\nu(l))}{\tilde{\beta}(\nu(l))} - \frac{\alpha(\nu(l))}{\beta(\nu(l))} \right) \ \frac{\beta(\nu(l)) \tilde{\beta}(\nu(l))}{r^{2\nu(l)}\ (\nu(l)+1)^2}.
\end{equation}
Then using (\ref{interpolation}), (\ref{phaseshifts}) and Proposition \ref{equivJost}, one obtains:

\begin{eqnarray*}
\mid f(l) \mid &\leq & C\ \mid e^{2i\delta_l} - e^{2i\tilde{\delta}_l}\mid \ \left| \frac{ \beta_0 (\nu(l))}{l\ r^{l}} \right|^2  , \\
&= &  \ o \left( \frac{1}{l^{n-3}} \  \left(\frac{ae}{2l}\right)^{2l} \right) \ \left| \frac{ \beta_0 (\nu(l))}{l\ r^{l}} \right|^2,
\end{eqnarray*}
thanks to our main hypothesis on the phase shifts. But, Stirling's formula and (\ref{betazero}) imply:
\begin{equation}
\beta_0(\nu(l)) \ =\ O \left(l^{\frac{n-1}{2}}\ \left( \frac{2l}{e} \right)^l \right) \ \ ,\ \ l \rightarrow +\infty.
\end{equation}
We deduce that:
for $r \geq a$, $f(l) \rightarrow 0$ as $l \rightarrow +\infty$. So, Cartwright's Theorem (with $B=0$), implies that $f(\nu) \rightarrow 0$ as $\nu \rightarrow +\infty$, which in turn implies the same result for $G(\nu)$.

\newpage

\appendix

\section{Appendix.}

\subsection{Some basic facts on the Bessel functions.}

First, let us recall some well-known definitions for the Bessel functions. We refer the reader to (\cite{Leb}, Chapter 5), or to the classic treatise by Watson \cite{Wa} to which we will make frequent references.

\vspace{0.5cm}\noindent
The Bessel function $J_{\nu}(z)$ is defined for $\nu \in \C$ and $\mid Arg\ z \mid < \pi$ by:
\begin{equation}\label{serieBessel}
J_{\nu}(z) = \sum_{k=0}^{+\infty} \frac{(-1)^k \left(\frac{z}{2}\right)^{\nu+2k}}{\Gamma(k+1) \Gamma(k+\nu+1)}.
\end{equation}

\vspace{0.5cm}\noindent
The Bessel functions of the third kind or Hankel functions, denoted by $H_{\nu}^{(1)} (z)$ and $H_{\nu}^{(2)} (z)$ are defined in terms of the Bessel functions  of the first and second kind by:
\begin{equation}\label{defHankel}
H_{\nu}^{(1)} (z) = J_{\nu}(z) + i Y_{\nu}(z)\ ,\ H_{\nu}^{(2)} (z) = J_{\nu}(z) - i Y_{\nu}(z),
\end{equation}
and can be written as:
\begin{eqnarray}\label{ecritureHankel}
H_{\nu}^{(1)} (z) &=& \frac{J_{-\nu}(z) - e^{-i\pi \nu}J_{\nu}(z)}{i\sin \nu \pi}, \\
H_{\nu}^{(2)} (z) &=& \frac{e^{i\pi \nu}J_{\nu}(z)-J_{-\nu}(z)}{i\sin \nu \pi}.
\end{eqnarray}

\vspace{0.2cm}\noindent
The Bessel functions $J_{\nu}(z)$ and the Hankel functions $H_{\nu}^{(j)} (z)$ are entire functions of $\nu$. Moreover, we have the following relations (see \cite{MOS}, p. 66):
\begin{equation}\label{conjugue}
\overline{J_{\nu} (z)}= J_{\bar{\nu}}  (\bar{z}) \ , \ \overline{Y_{\nu} (z)}= Y_{\bar{\nu}} (\bar{z})\ ,\ \overline{H_{\nu}^{(1)} (z)}= H_{\bar{\nu}}^{(2)} (\bar{z}).
\end{equation}
\begin{equation}\label{pariteHankel}
H_{-\nu}^{(1)} (z) = e^{i\nu \pi} H_{\nu}^{(1)} (z) \ ,\ H_{-\nu}^{(2)} (z) = e^{-i\nu \pi} H_{\nu}^{(2)} (z).
\end{equation}

\subsection{Estimates on the imaginary axis for the Hankel functions.}

In this Section, we shall give useful estimates for $H_{iy}^{(1)} (r)$ and $H_{iy}^{(2)} (r)$ with respect to $r$ and $y$.  These results are probably well-known, but as we were unable to find a precise reference, we will give the simple proofs below.

\begin{prop}\label{imaxis}
\hfill\break
For any $r>0$ and $y \in \R^*$, one has:
\begin{eqnarray}\label{Hankel1}
&&\mid H_{iy}^{(1)} (r) \mid  \leq \frac{2^{\frac{3}{4}}}{\sqrt{\pi r}} \ e^{\pd y} \ \ , \ \
\mid H_{iy}^{(2)} (r) \mid  \leq   \frac{2^{\frac{3}{4}}} {\sqrt{r}} \ e^{-\pd y}. \\
&&\mid H_{iy}^{(1)} (r) \mid  \leq \frac{2}{\sqrt{\pi}} \ (r \mid y \mid)^{-\frac{1}{4}} \ e^{\pd y} \ \  , \ \
\mid H_{iy}^{(2)} (r) \mid  \leq   \frac{2}{\sqrt{\pi}} \ (r \mid y \mid)^{-\frac{1}{4}} \ e^{-\pd y}.
\end{eqnarray}
\end{prop}

\begin{proof}
By (\ref{conjugue}), it suffices to estimate $\mid H_{iy}^{(1)} (r) \mid$. We write
\begin{equation}
\mid H_{iy}^{(1)} (r) \mid^2 = H_{iy}^{(1)} (r) \ \overline{H_{iy}^{(1)} (r)} = H_{iy}^{(1)} (r) H_{-iy}^{(2)} (r) =
e^{\pi y} H_{iy}^{(1)} (r) H_{iy}^{(2)} (r),
\end{equation}
where we have used (\ref{conjugue}) and (\ref{pariteHankel}). Hence, using (\ref{defHankel}), we obtain:
\begin{eqnarray}\label{moduleHankel}
\mid H_{iy}^{(1)} (r) \mid^2  &=& e^{\pi y} \ \left( J_{iy}(r) + i Y_{iy}(r) \right) \left( J_{iy}(r) - i Y_{iy}(r) \right), \nonumber \\
                              &=&  e^{\pi y} \  \left( J_{iy}^2(r) +  Y_{iy}^2(r) \right).
\end{eqnarray}
Now, for $\Re z >0$, we recall the Nicholson's integral representation  (see \cite{MOS}, p. 93):
\begin{equation}\label{Nicholson}
J_{\nu}^2(z) +  Y_{\nu}^2(z) = \frac{8}{\pi^2} \ \int_0^{+\infty} K_0 (2z \sinh t) \ \cosh (2\nu t) \ dt,
\end{equation}
where $K_0(z)$ is the Macdonald's function given by (see \cite{Ol}, Eq. 10.32.9):
\begin{equation}\label{Macdonald}
K_0(z) = \int_0^{+\infty} e^{-z \cosh t} \ dt \ , \ \Re z >0.
\end{equation}
It follows from (\ref{moduleHankel}) and (\ref{Nicholson}) that for all $r>0$:
\begin{equation}\label{Nicholson1}
\mid H_{iy}^{(1)} (r) \mid^2 = \frac{8e^{\pi y}}{\pi^2}\ \int_0^{+\infty} K_0 (2r \sinh t) \ \cos (2 ty) \ dt.
\end{equation}
Clearly, from (\ref{Macdonald}), we see that the Macdonald's function $K_0(x)$ is a positive decreasing function for $x>0$, and using the inequality $\cosh t \geq 1+ \frac{t^2}{2}$, we obtain easily for $x>0$:
\begin{equation}\label{estMD}
K_0(x) \leq \sqrt{\frac{\pi}{2x}} \ e^{-x}.
\end{equation}


\noindent
Then, since $\sinh t \geq t$ for $t \geq 0$, it follows from (\ref{Nicholson1}) and (\ref{estMD}) that :
\begin{eqnarray*}\label{Nicholson2}
\mid H_{iy}^{(1)} (r) \mid^2 & \leq & \frac{8e^{\pi y}}{\pi^2}\ \int_0^{+\infty} K_0 (2r  t)  \ dt \\
                             & \leq &  \frac{8e^{\pi y}}{\pi^2}\ \int_0^{+\infty}  \sqrt{\frac{\pi}{4rt}} \ e^{-2rt}  \ dt \\
                             &\leq &\frac{2\sqrt{2}}{\pi r} \ e^{\pi y},
\end{eqnarray*}
which proves the first part of the Proposition.

\vspace{0.5cm}
Now, assume for instance that $y>0$. Making the change of variables $ s=2ty$ in (\ref{Nicholson1}), we obtain:
\begin{equation}\label{Nicholson3}
\mid H_{iy}^{(1)} (r) \mid^2 = \frac{4e^{\pi y}}{y\pi^2}\ \int_0^{+\infty} f(s) \ \cos (s) \ ds,
\end{equation}
where ${\displaystyle{f(s) = K_0\left( 2r \sinh (\frac{s}{2y}) \right)}}$ is a decreasing function for $s>0$. We write this later integral integral as:
\begin{eqnarray*}
\int_0^{+\infty} f(s) \ \cos (s) \ ds = &&\int_0^{\pd} f(s) \ \cos (s) \ ds \\
&&+ \sum_{n=0}^{\infty}\  \left(
\int_{\pd+2n\pi} ^{\pd + (2n+1)\pi} f(s) \ \cos (s) \ ds + \int_{\pd+(2n+1)\pi} ^{\pd + (2n+2)\pi} f(s) \ \cos (s) \ ds \right).
\end{eqnarray*}
By the first mean value theorem, there exists $a_n \in [\pd + 2n\pi, \pd + (2n+1)\pi]$ and $b_n \in [\pd + (2n+1)\pi, \pd + (2n+2)\pi]$ such that:
\begin{equation}
\int_{\pd+2n\pi} ^{\pd + (2n+1)\pi} f(s) \ \cos (s) \ ds + \int_{\pd+(2n+1)\pi} ^{\pd + (2n+2)\pi} f(s) \ \cos (s) \ ds = -2 (f(a_n) -f(b_n)) \leq0.
\end{equation}
It follows from (\ref{Nicholson3}) that
\begin{eqnarray*}\label{Nicholson4}
\mid H_{iy}^{(1)} (r) \mid^2 &\leq&  \frac{4e^{\pi y}}{y\pi^2}\ \int_0^{\pd} f(s) \ \cos (s) \ ds \\
&\leq&  \frac{4e^{\pi y}}{y\pi^2}\ \int_0^{\pd} K_0\left( 2r \sinh (\frac{s}{2y}) \right)   \ ds \\
&\leq& \frac{4e^{\pi y}}{y\pi^2}\ \int_0^{\pd} K_0 (\frac{rs}{y})  \ ds,
\end{eqnarray*}
where we have still used that  $K_0(x)$ is a decreasing function and $\sinh x \geq x$ for $x \geq 0$. Then, the result comes from immediately from (\ref{estMD}).
\end{proof}

\subsection{A new integral representation for the product of two Bessel functions.}

In the next Theorem, we give an integral representation formula for the product of the Bessel function $J_{\nu}(r)$ and the Hankel function $H_{\nu}^{(1)}(R)$. To our knowledge, this result seems to be new and will be very useful to estimate the Green kernel $K(r,s, \nu)$ appearing in Proposition \ref{newrep}.

\begin{thm}\label{newrel}
\hfill\break
For $\Re \nu >0$ and $0<r< R$, one has the following integral representation:
\begin{equation}\label{intproduct}
J_{\nu}(r) \ H_{\nu}^{(1)}(R) = -\frac{2i}{\pi}  \ \int_0^{+\infty} e^{i(r+R) \cosh x} \ J_{2\nu} (2 \sqrt{rR} \ \sinh x) \ dx .
\end{equation}
\end{thm}

\begin{proof}
We start from the integral relation due to H. Buchholz for the product of two (normalized) Whittaker functions, (see \cite{Bu}, p. 86, Eq. (5c), or \cite{GrRy}, BU 86 (5c), p. 716, but we warn the reader of a misprint in \cite{GrRy}):
\begin{eqnarray}\label{buchholz}
\int_0^{+\infty}  e^{-\frac{t}{2}(a_1 +a_2)\cosh x}\ \left( \coth \left( \frac{x}{2} \right) \right)^{2k} \hspace{-0.7cm}&&\ I_{2\nu} (t\sqrt{a_1 a_2} \sinh x) \ dx = \nonumber \\
&&\frac{\Gamma(\half+\nu-k)}{t\sqrt{a_1 a_2}\ \Gamma(1+2\nu)}\ W_{k, \nu}(ta_1)\  M_{k,\nu}(ta_2 ),
\end{eqnarray}
for $\Re (\half +\nu-k)>0, \ \Re \nu >0, \ a_1 >a_2>0, \ t>0$, where $W_{k, \nu}(z), \ M_{k,\nu}(z )$ are the Whittaker functions, (see \cite{MOS}, p. 295), and
$I_{\nu}(z)$ is the modified Bessel function which is related to Bessel function $J_{\nu}(z)$ by the formula (\cite{Leb}, Eq. (5.7.4)):
\begin{equation}\label{besselmod}
I_{\nu}(z) = e^{-i\pd \nu} \ J_{\nu}(iz) \ ,\ -\pi < Arg\ z <\pd .
\end{equation}

\vspace{0.2cm}\noindent
We recall, (see \cite{MOS}, p. 305), that: $\forall z \in \C \backslash \R^-$,
\begin{eqnarray}
W_{0, \nu}(z) &=& \frac{i}{2} \sqrt{\pi z} \ e^{-i\pd \nu} \ H_{\nu}^{(1)} (\frac{iz}{2}), \\
M_{0, \nu}(z) &=& \Gamma(1+\nu) \ 2^{2\nu} \ e^{-i\pd \nu} \ \sqrt{ z}\  J_{\nu} (\frac{iz}{2}).
\end{eqnarray}
In the equation (\ref{buchholz}), we take $a_1 =2R, \ a_2 =2r, \ k=0$ and we obtain easily:
\begin{equation*}
J_{\nu}(itr)\ H_{\nu}^{(1)}(itR) = -\frac{2i}{\sqrt{\pi}} \ \frac{\Gamma(2\nu+1)}{\Gamma(\nu+\half)\Gamma(\nu+1) 2^{2\nu}}\ e^{i\nu \pi} \ \int_0^{+\infty} \
e^{-t(r+R)\cosh x}\ I_{2\nu} (2t\sqrt{rR} \ \sinh x) \ dx.
\end{equation*}
Now, using the duplication formula for the Gamma function (\cite{Leb}, Eq. (1.2.3)):
\begin{equation}
\Gamma(2\nu +1) = \frac{1}{\sqrt{\pi}} \ 2^{2\nu} \ \Gamma (\nu+\half) \Gamma(\nu+1),
\end{equation}
we obtain immediately:
\begin{equation}\label{treel}
J_{\nu}(itr)\ H_{\nu}^{(1)} (itR) = -\frac{2i}{\pi} \ e^{i\nu \pi} \ \int_0^{+\infty} \ e^{-t(r+R)\cosh x}\ I_{2\nu} (2t\sqrt{rR} \ \sinh x) \ dx.
\end{equation}
Now, we see it  is easy to  extend (\ref{treel}) for $\Re t \geq 0$, $t \not=0$ recalling that, (see \cite{Leb}, eqs. (5.7.1) and (5.11.8)),
\begin{eqnarray}
I_{2\nu}(z) &\sim& \frac{1}{\Gamma (2\nu+1)} \ \left( \frac{z}{2}\right)^{2\nu} \  ,\  z \rightarrow 0, \ \ \mid Arg\ z \mid <\pi, \\
I_{2\nu}(z) &\sim& \frac{1}{\sqrt{2\pi z}} \left( e^z + e^{-z \pm i\pi(2\nu+\half)}\right) \  ,\  z \rightarrow \infty, \ \ \mid Arg\ z \mid <\pi-\delta,\ \pm \Im z >0.
\end{eqnarray}
Then taking $t=-i$ in (\ref{treel}) and using (\ref{besselmod}), we obtain the result.

\end{proof}

\subsection{Uniform estimate for the Green Kernel $K(r,s,\nu)$.}

In this subsection, we use Theorem \ref{newrel} to prove an uniform estimate with respect to $r,s>0$ and $\Re \nu \geq 0$ for the Green kernel defined in Section 3 by $K(r,s , \nu) = u(s) v(r)$ if $s \leq r$ and $K(r,s , \nu) = u(r) v(s)$ if $s \geq r$, where
\begin{equation*}
u(r) = \sqrt{\frac{\pi r}{2}} \ J_{\nu} (r) \ \ ,\ \ v(r) = -i \sqrt{\frac{\pi r}{2}} \ H_{\nu}^{(1)} (r).
\end{equation*}
\vspace{0.2cm}\noindent
Let us begin by an elementary result:

\begin{lemma}\label{intmodule}
\hfill\break
For $2\Re \nu +1 >\delta >0$, one has:
\begin{equation*}
\int_0^{+\infty} \frac{\mid J_{\nu}(t) \mid^2}{t^{\delta}} \ dt = \frac {\Gamma(\delta) \ \Gamma (\Re \nu + \frac{1-\delta}{2})}{2^{\delta} \mid \Gamma( \frac{\delta+1}{2} + i \Im \nu)\mid^2 \ \Gamma(\Re \nu + \frac{1+\delta}{2})}.
\end{equation*}
\end{lemma}

\begin{proof}
We recall, (\cite{Wa}, p. 403), that for $\Re(\mu+\nu+1) > \Re \delta >0$, one has:
\begin{equation}\label{integralbessel}
\int_0^{+\infty} \frac{J_{\mu}(t) J_{\nu}(t)}{t^{\delta}} \ dt = \frac{\Gamma(\delta) \ \Gamma( \frac{\mu+\nu-\delta+1}{2})}
{2^{\delta} \Gamma( \frac{\delta+\nu-\mu+1}{2})\ \Gamma( \frac{\delta+\nu+\mu+1}{2})\ \Gamma( \frac{\delta-\nu+\mu+1}{2})}.
\end{equation}
We choose $\mu=\bar{\nu}$, $\delta>0$ in (\ref{integralbessel}), and taking account $J_{\bar{\nu}}(t)= \overline{J_{\nu}(t)}$, the lemma is proved.
\end{proof}

\vspace{0.2cm}\noindent
We deduce the following estimate:

\begin{coro}\label{majorationint}
\hfill\break
For all $\delta \in (0,1)$, there exists $C_{\delta}>0$ such that:
\begin{equation*}
\int_0^{+\infty} \frac{\mid J_{\nu}(t) \mid^2}{t^{\delta}} \ dt \ \leq C_{\delta} \ e^{\pi \mid \Im \nu \mid} \ (1+ \Re \nu)^{-\delta} \ ,\ \forall \Re \nu \geq 0 .
\end{equation*}
\end{coro}

\begin{proof}
Let us fix $\delta \in (0,1)$. By Lemma \ref{intmodule}, one has:
\begin{equation}\label{rappel}
\int_0^{+\infty} \frac{\mid J_{\nu}(t) \mid^2}{t^{\delta}} \ dt \ = \ \frac {\Gamma(\delta)}{2^{\delta}} \ \frac{\Gamma (\Re \nu + \frac{1-\delta}{2})}
{\Gamma(\Re \nu + \frac{1+\delta}{2})} \ \ \frac{1}{\mid \Gamma( \frac{\delta+1}{2} + i \Im \nu)\mid^2}.
\end{equation}
We recall that for $x \geq \half$ (see \cite{Ol}, Eq. (5.6.7)),
\begin{equation}\label{minorationGamma}
\mid \Gamma(x+iy) \mid \geq \frac{1}{\sqrt{\cosh (\pi y)}} \ \Gamma(x).
\end{equation}
Moreover, one has the following asymptotics (\cite{Leb}, p. 15),
\begin{equation}\label{asymptGamma}
\frac{\Gamma(z+\alpha)}{\Gamma (z+\beta)} \ \sim z^{\alpha-\beta} \ ,\ z \rightarrow + \infty.
\end{equation}
Hence, the corollary follows immediately from (\ref{rappel}), (\ref{minorationGamma}) and (\ref{asymptGamma}).
\end{proof}

\vspace{0.5cm}\noindent
Now, we can establish the following result:

\begin{thm}\label{estproduitbessel}
\hfill\break
For any  $\delta \in (0,1)$, there exists $C_{\delta}>0$ such that, for all $0<r\leq R$ and  $\Re \nu \geq 0$,
\begin{equation*}
\mid J_{\nu}(r) \ H_{\nu}^{(1)}(R) \mid \ \leq \ C_{\delta} \ e^{\pi \mid \Im \nu \mid} \ (1+\Re \nu)^{-\frac{\delta}{2}} \ (rR)^{\frac{\delta -1}{4}} .
\end{equation*}
\end{thm}

\begin{proof}
Of course, by a standard continuity argument, it suffices to prove the estimate for $0<r<R$ and $\Re \nu >0$. Using Theorem \ref{newrel}, one has:
\begin{equation}\label{inequalityproduct}
\mid J_{\nu}(r) \ H_{\nu}^{(1)}(R)\mid  \leq  \frac{2}{\pi}  \ \int_0^{+\infty}  \mid  J_{2\nu} (2 \sqrt{rR} \ \sinh x) \mid \  dx.
\end{equation}
We make the change of variables $s=2\sqrt{rR} \ \sinh x$ and we obtain:
\begin{equation}\label{inequalityproduct1}
\mid J_{\nu}(r) \ H_{\nu}^{(1)}(R)\mid  \leq  \frac{1}{\pi \sqrt{rR} }  \ \int_0^{+\infty}  \mid  J_{2\nu} (s) \mid \ \left(1+\frac{s^2}{4rR} \right)^{-\half} \   ds.
\end{equation}
We take $\delta \in (0,1)$ and by the Cauchy-Schwarz inequality, one deduces from (\ref{inequalityproduct1}):
\begin{eqnarray}\label{inequalityproduct2}
\mid J_{\nu}(r) \ H_{\nu}^{(1)}(R)\mid  &\leq&  \frac{1}{\pi \sqrt{rR} }  \ \left( \int_0^{+\infty}  \frac{\mid  J_{2\nu} (s) \mid^2}{s^{\delta}} \ ds \right)^{\half}
\nonumber \\
&& \cdot \ \left( \int_0^{+\infty} s^{\delta} \ \left( 1+\frac{s^2}{4rR} \right)^{-1} \ ds \right)^{\half} .
\end{eqnarray}
By Corollary \ref{majorationint},
\begin{equation}
\left( \int_0^{+\infty}  \frac{\mid  J_{2\nu} (s) \mid^2}{s^{\delta}} \ ds \right)^{\half} \ \leq \ C_{\delta} \ e^{\pi \mid \Im \nu \mid } \ (1+\Re \nu)^{-\frac{\delta}{2}}.
\end{equation}
Hence, the Theorem follows from the obvious inequality:
\begin{equation}
\int_0^{+\infty} s^{\delta} \ \left( 1+\frac{s^2}{4rR} \right)^{-1} \ ds \ \leq \ C_{\delta} \ (rR)^{\frac{\delta +1}{2}}.
\end{equation}
\end{proof}

\vspace{0.5cm}\noindent
From the definition of the Green kernel $K(r,s, \nu)$, we deduce an uniform bound with respect to $r,s$ and $\nu$ in the right complex half-plane :

\begin{coro}\label{estGreenK}
\hfill\break
For any  $\delta \in (0,1)$, there exists $C_{\delta}>0$ such that, for all $r,s>0$ and  $\Re \nu >0$,
\begin{equation*}
\mid K(r,s,\nu) \mid \ \leq \ C_{\delta} \ e^{\pi \mid \Im \nu \mid} \ (1+\Re \nu)^{-\frac{\delta}{2}} \ (rs)^{\frac{\delta +1}{4}} .
\end{equation*}
\end{coro}

\subsection{Uniform asymptotics for the Bessel functions with respect to the order.}

In this section, we shall recall some  uniform asymptotics for the Bessel function $J_{\nu}(r)$ and the Hankel function $H_{\nu}^{(1)}$ with respect to $\nu$ when $r$ belongs to a {\it{compact set}}. We emphasize that all these uniform asymptotics  fail if $r \in (0, +\infty)$.

\begin{prop}\label{unifBessel}
\hfill\break
Let $\delta>0$ be small enough. For $r>0$ belonging to a compact set, we have the uniform asymptotics when $\nu \rightarrow \infty$:
\begin{eqnarray*}
J_{\nu}(r) &=& \frac { \left(\frac{r}{2}\right)^{\nu} }{\Gamma(\nu +1) } \ \left( 1 + O\left(\frac{1}{\nu}\right) \right) \ \ , \ \mid Arg \ \nu \mid \leq \pi -\delta. \\
H_{\nu}^{(1)} (r) &=& -\frac{i}{\pi} \left( \frac{r}{2} \right)^{-\nu} \ \Gamma(\nu) \ \left( 1 + O\left(\frac{1}{\nu}\right) \right)  \ \ , \ \mid Arg \ \nu \mid \leq \pd -\delta.
\end{eqnarray*}

\end{prop}

\begin{proof}
The first asymptotics follows directly from (\ref{serieBessel}). We refer to (\cite{Olver}, p. 374) for the same explanation for the modified Bessel function $I_{\nu}(z)$. To prove the second asymptotics, we use the following result (\cite{SH}, Eq. (1.2)) for the Macdonald's function $K_{\nu}(z)$:
\begin{equation}
K_{\nu}(z) = \frac{1}{2} \ \frac{\Gamma(\nu)}{(\frac{z}{2})^{\nu}} \  \left(1 + O\left(\frac{1}{\nu}\right)\right) , \ \mid Arg \ \nu \mid \leq \pd -\delta.
\end{equation}
We conclude using the relation (\cite{Leb}, Eq. (5.7.1)):
\begin{equation}
H_{\nu}^{(1)}(r) = -\frac{2i}{\pi} \ e^{-i \nu \pd} \ K_{\nu}(-ir).
\end{equation}
\end{proof}

\vspace{0.2cm}\noindent
\begin{rem}\label{uniformeprecis0}
 A more precise uniform asymptotic expansion for $J_{\nu}(r)$ and $H_{\nu}^{(1)}(r)$ for $r$ in a compact set is necessary for the study of the Regge poles. Following \cite{SH}, one has:

 \begin{eqnarray}
J_{\nu}(r) &=& \frac { \left(\frac{r}{2}\right)^{\nu} }{\Gamma(\nu +1) } \ \left( 1 - \frac{r^2}{4\nu}+ O\left(\frac{1}{\nu^2}\right) \right) \ \ , \ \mid Arg \ \nu \mid \leq \pi -\delta. \\
H_{\nu}^{(1)} (r) &=& -\frac{i}{\pi} \left( \frac{r}{2} \right)^{-\nu} \ \Gamma(\nu) \ \left( 1+\frac{r^2}{4\nu} + O\left(\frac{1}{\nu^2}\right) \right)  \ \ , \ \mid Arg \ \nu \mid \leq \pd -\delta.
\end{eqnarray}

\end{rem}

We deduce from the previous remark the following result:

\vspace{0.2cm}
\begin{coro}\label{uniformeprecis}
\hfill\break
For $Arg \ \nu \in [\pd -\delta, \pd]$ and for $r,s$ in a compact set, one has the following uniform asymptotics
when $\nu \rightarrow \infty$:
\begin{equation*}
J_{\nu}(s) \ H_{\nu}^{(1)} (r) =  \frac{1}{i\pi \nu} \left( \frac{s}{r}\right)^{\nu} \ \left(1+\frac{r^2 -s^2}{4\nu} +O\left(\frac{1}{\nu^2}\right)\right)  + \frac{2}{\Gamma^2 (\nu+1)} \left(\frac{rs}{4}\right)^{\nu} \left(1- \frac{r^2 +s^2}{4\nu} +O\left(\frac{1}{\nu^2}\right) \right) .
\end{equation*}
\end{coro}

\begin{proof}
Using (\ref{ecritureHankel}), one has:
\begin{equation}\label{dev}
J_{\nu}(s) H_{\nu}^{(1)}(r) = \frac{ J_{\nu}(s) J_{-\nu}(r) - e^{-i\pi \nu} \ J_{\nu}(s)J_{\nu}(r)}{i\sin(\pi \nu)}.
\end{equation}
Now, we remark that $\mid Arg \ (-\nu )\mid \leq \pi -\delta$. Then, we can also use the previous uniform asymptotics for $J_{-\nu}(r)$. For instance, we have as $\nu \rightarrow \infty$:
\begin{equation}
J_{\nu}(s) J_{-\nu}(r) = \frac{ \left( \frac{s}{r}\right)^{\nu}}{\Gamma(\nu+1)\Gamma(-\nu+1)}\ \left(1+\frac{r^2 -s^2}{4\nu} +O\left(\frac{1}{\nu^2}\right)\right) .
\end{equation}
But by the complement formula for the Gamma function (\cite{Leb}, Eq. (1.2.2)):
\begin{equation}
\Gamma(\nu)\Gamma(1-\nu) = \frac{\pi}{\sin(\pi \nu)},
\end{equation}
one obtains :
\begin{equation}
\frac{ J_{\nu}(s) J_{-\nu}(r)}{i\sin(\pi \nu)} =  \frac{1}{i\pi \nu} \left( \frac{s}{r}\right)^{\nu}\ \left(1+\frac{r^2 -s^2}{4\nu} +O\left(\frac{1}{\nu^2}\right)\right) .
\end{equation}
We can study the second term in (\ref{dev}) similarly.
\end{proof}

\subsection{The Born approximation.}

In this section, we prove that for a large set of potentials $q(r)$ belonging to the class $\mathcal{A}$, the restriction of the Fourier transform of the potential on any ball determines uniquely $q(r)$. This result is coherent with the Born approximation. We recall also the following facts :  for general potentials $V(x)$ (not necessary with spherical symmetry) satisfying for all $x \in \R^n$,
 \begin{equation}\label{symbole}
\mid  V(x) \mid \ \leq C \ (1+\mid x\mid)^{-\rho} \ \ ,\ \ \rho>n,
\end{equation}
it is shown in \cite{HeNo} that the scattering matrix $S(\mu)$, $\mu \in [\lambda , \lambda +\delta ]$, where $\lambda$ is a fixed energy and $\delta>0$ is arbitrary small, determines the Fourier transform of $V$ on the ball $B(0,2\sqrt{\lambda})$, and in \cite{Is}, this result is extended to the case $\rho > \frac{3}{2}$ for smooth potentials.

\vspace{0.5cm}\noindent
We have the following theorem:

\begin{thm}\label{UniqueFourier}
\hfill\break
Let $V(x)$  be a central potential in the class $\mathcal{A}$ with $\rho>n$.
Then, the restriction of the Fourier transform of $V(x)$ on any ball determines uniquely the potential.
\end{thm}

\begin{proof}
So, let us assume that $\hat{V}(\xi)$ is known for $|\xi| \leq a$ for some $a>0$. We write:
\begin{eqnarray*}
\hat{V}(\xi) &=& \frac{1}{(2\pi)^{\frac{n}{2}}} \ \int_{\R^n} e^{-ix \cdot \xi} \ V(x) \ dx \\
           &=& \frac{1}{(2\pi)^{\frac{n}{2}}} \ \int_0^{+\infty} \left( \int_{\mathbb{S}^{n-1}}  e^{-ir\xi \cdot \omega} \ d\omega \right) \ V(r) \ r^{n-1} \ dr.
\end{eqnarray*}
But, it is well-known that for $\xi \not=0$, (see for instance \cite{BaCa}, Eq. $(3.2)$):
\begin{equation}
\int_{\mathbb{S}^{n-1}}  e^{-ir\xi \cdot \omega} \ d\omega = \frac{2\pi}{|r\xi|^{\frac{n}{2}-1}} \ J_{\frac{n}{2}-1} (r \mid \xi \mid).
\end{equation}
By our hypothesis, $q(r)$ belongs to the Hardy class $H_+^2$, so the Paley-Wiener theorem (see (\ref{PaleyWiener})) asserts that there exists $f \in L^2(0, +\infty)$ such that:
\begin{equation}
V(r) = \int_0^{+\infty} e^{-tr} \ f(t)\ dt
\end{equation}
Then, we have for $\xi \not=0$:
\begin{equation}
\hat{V}(\xi) = \frac{1}{\mid \xi \mid^{\frac{n}{2}-1}} \ \int_0^{+\infty} \left( \int_0^{+\infty} e^{-tr} J_{\frac{n}{2}-1} (r \mid \xi \mid) \ r^{\frac{n}{2}} \ dr \right) f(t) \ dt.
\end{equation}
Recalling that (\cite{Wa}, Eq. $(6)$, p. 386):
\begin{equation}
\int_0^{+\infty} e^{-\alpha r} J_{\nu} (\beta r)\  r^{\nu+1} \ dr = \frac{ 2\alpha \ (2\beta)^{\nu}\  \Gamma(\nu+\frac{3}{2})}   {\sqrt{\pi} \ (\alpha^2 +\beta^2)^{\nu + \frac{3}{2}}} \ \ ,\ \ \Re \nu >-1,\ \Re \alpha > |\Im \beta |,
\end{equation}
we obtain easily for $\xi \not =0$:
\begin{equation}
\hat{V}(\xi) = \frac{2^{\frac{n}{2}} \Gamma( \frac{n+1}{2})} {\sqrt{\pi} } \ \int_0^{+\infty} \frac{t}{(t^2 +|\xi|^2)^{\frac{n+1}{2}}} \ f(t) \ dt.
\end{equation}
Clearly, this later integral is analytic with respect to $\mid \xi \mid$ in $(0,+\infty)$, so using the standard analytic continuation principle and the inverse Fourier transform, Theorem \ref{UniqueFourier} is proved.

\end{proof}

\vspace{0.5cm}\noindent
{\it{Acknowledgments.}}
\par\noindent
Both authors would like to warmly thank Marco Marletta and Roman Novikov for useful discussions on the Regge poles and inverse scattering theory.


\end{document}